\ifdefined\isfocs
\documentclass[10pt, conference, compsocconf]{IEEEtran}
\else
\documentclass[11pt]{article}
\usepackage[margin=1in]{geometry}
\fi

\usepackage[utf8]{inputenc}
\usepackage[T1]{fontenc}

\usepackage{caption}
\usepackage{thm-restate}
\usepackage{bm}
\usepackage{bbm}
\usepackage{mathrsfs}           % \mathscr
\usepackage{tikz}
\usetikzlibrary{positioning,chains,fit,shapes,calc}
\usepackage{algorithm}
\usepackage{algpseudocode}
\usepackage{multirow}
\usepackage{enumerate}

\usepackage{amsmath,amssymb, amsthm}    % AMS
\usepackage{verbatim}           % comment environment
\usepackage{graphicx,float}     % figures
\usepackage{ifthen,calc}        % macros
\usepackage{textcomp}           % \textcent
\usepackage{hhline}             % \hhline
\usepackage{float}              % \float
\numberwithin{table}{section}

\usepackage[vflt]{floatflt}     % default is vflt
\usepackage{color}
\usepackage[colorlinks,linkcolor=blue,citecolor=red]{hyperref}
\newtheorem{theorem}{Theorem}[section]
\newtheorem{claim}[theorem]{Claim}

\newtheorem{definition}[theorem]{Definition}
\newtheorem{remark}[theorem]{Remark}
\newtheorem{lemma}[theorem]{Lemma}
\newtheorem{fact}[theorem]{Fact}

\newcommand\numberthis{\addtocounter{equation}{1}\tag{\theequation}}

\newcommand{\Pdn}{\mathbb{S}^n_{> 0}}

\usepackage{url}

\newcommand{\R}{\mathbb{R}}
\newcommand{\wt}{\widetilde}
\newcommand{\wh}{\widehat}
\newcommand{\new}{\mathrm{new}}
\newcommand{\tr}{\mathrm{tr}}
\newcommand{\diag}{\mathrm{diag}}
\newcommand{\nnz}{\mathrm{nnz}}
\newcommand{\op}{\mathrm{op}}

\newcommand{\poly}{\mathrm{poly}}
\newcommand{\rank}{\mathrm{rank}}
\newcommand{\OPT}{\mathsf{OPT}}

\newcommand{\cost}[1]{\mathcal{T}_{#1}}
\newcommand{\norm}[1]{\left\| #1 \right\|}

\newcommand{\Tmat}{\mathcal{T}_{\mathrm{mat}}}
\newcommand{\vect}{\mathsf{vec}}

\definecolor{mygreen}{RGB}{80,180,0}
\newcommand{\eat}[1]{}
\newcommand{\Zhao}[1]{{\color{mygreen}[Zhao: #1]}}

\newcommand{\abs}[1]{|#1|}

\title{
A Faster Interior Point Method for Semidefinite Programming
} 
\date{}
\ifdefined\isfocs
\author{
\IEEEauthorblockN{Haotian Jiang}
\IEEEauthorblockA{
\texttt{jhtdavid@uw.edu}\\
 University of Washington}
\and
\IEEEauthorblockN{Tarun Kathuria}
\IEEEauthorblockA{
\texttt{tarunkathuria@berkeley.edu}\\
University of California, Berkeley}
\and
\IEEEauthorblockN{Yin Tat Lee}
\IEEEauthorblockA{
\texttt{yintat@uw.edu}\\
University of Washington, 
Microsoft Research Redmond}
\and
\IEEEauthorblockN{Swati Padmanabhan}
\IEEEauthorblockA{
\texttt{pswati@uw.edu} \\
University of Washington}
\and
\IEEEauthorblockN{Zhao Song}
\IEEEauthorblockA{
\texttt{magic.linuxkde@gmail.com} \\
Columbia University, Princeton University, and
 Institute for Advanced Study}
}
\else
\author{
			Haotian Jiang\thanks{\texttt{jhtdavid@uw.edu}. University of Washington.} 
			\and
			Tarun Kathuria\thanks{\texttt{tarunkathuria@berkeley.edu}. University of California, Berkeley.}
			\and
			Yin Tat Lee\thanks{\texttt{yintat@uw.edu}. University of Washington.}
			\and
			Swati Padmanabhan\thanks{\texttt{pswati@uw.edu}. University of Washington}
			\and
			Zhao Song\thanks{\texttt{zhaos@ias.edu}. Princeton University and Institute for Advanced Study.} 
}
\fi

\begin{document}

\ifdefined\isfocs
\maketitle
\begin{abstract}

Semidefinite programs (SDPs) are a fundamental class of optimization problems with important recent applications in approximation algorithms, quantum complexity, robust learning, algorithmic rounding, and adversarial deep learning. 
This paper presents a faster interior point method to solve generic SDPs with variable size $n \times n$ and $m$ constraints in time 
\begin{align*}
\widetilde{O}(\sqrt{n}( mn^2 + m^\omega + n^\omega) \log(1 / \epsilon) ),
\end{align*} 
where $\omega$ is the exponent of matrix multiplication and $\epsilon$ is the relative accuracy. In the predominant case of $m \geq n$, our runtime outperforms that of the previous fastest SDP solver, which is based on the cutting plane method \cite{jlsw20}. 

Our algorithm's runtime can be naturally interpreted as follows: $\widetilde{O}(\sqrt{n} \log (1/\epsilon))$ is the number of iterations needed for our interior point method, $mn^2$ is the input size, and $m^\omega + n^\omega$ is the time to invert the Hessian and slack matrix in each iteration.  
These constitute natural barriers to further improving the runtime of interior point methods for solving generic SDPs.

%%% Zhao : I don't feel need the following paragraphs.
%To achieve this runtime, we developed a new potential function that captures the low-rank updates of matrices across all iterations.  Unlike recent optimization algorithms that leverage fast rectangular matrix multiplication results in a black-box manner, we use the tensor rank technique [Coppersmith and Winograd, STOC 1987] from the area of matrix multiplication. In doing so, we prove several non-trivial, algorithm-driven results about fast rectangular matrix multiplication, which might be of independent interest.

\end{abstract}

\else
\begin{titlepage}
\maketitle
\begin{abstract}

Semidefinite programs (SDPs) are a fundamental class of optimization problems with important recent applications in approximation algorithms, quantum complexity, robust learning, algorithmic rounding, and adversarial deep learning. 
This paper presents a faster interior point method to solve generic SDPs with variable size $n \times n$ and $m$ constraints in time 
\begin{align*}
\widetilde{O}(\sqrt{n}( mn^2 + m^\omega + n^\omega) \log(1 / \epsilon) ),
\end{align*} 
where $\omega$ is the exponent of matrix multiplication and $\epsilon$ is the relative accuracy. In the predominant case of $m \geq n$, our runtime outperforms that of the previous fastest SDP solver, which is based on the cutting plane method \cite{jlsw20}. 

Our algorithm's runtime can be naturally interpreted as follows: $\widetilde{O}(\sqrt{n} \log (1/\epsilon))$ is the number of iterations needed for our interior point method, $mn^2$ is the input size, and $m^\omega + n^\omega$ is the time to invert the Hessian and slack matrix in each iteration.  
These constitute natural barriers to further improving the runtime of interior point methods for solving generic SDPs.

%%% Zhao : I don't feel need the following paragraphs.
%To achieve this runtime, we developed a new potential function that captures the low-rank updates of matrices across all iterations.  Unlike recent optimization algorithms that leverage fast rectangular matrix multiplication results in a black-box manner, we use the tensor rank technique [Coppersmith and Winograd, STOC 1987] from the area of matrix multiplication. In doing so, we prove several non-trivial, algorithm-driven results about fast rectangular matrix multiplication, which might be of independent interest.

\end{abstract}
\thispagestyle{empty}
\end{titlepage}

\newpage 

{\hypersetup{linkcolor=black}
\tableofcontents
}

\newpage
\fi

\section{Introduction}
\label{sec:intro}

Semidefinite programs (SDPs) constitute a class of convex optimization problems that optimize a linear objective  over the intersection of the cone of positive semidefinite matrices with an affine space. 
SDPs generalize linear programs and have a plethora of applications in operations research, control theory, and theoretical computer science~\cite{vb96}. Applications in theoretical computer science include improved approximation algorithms for fundamental problems (e.g., Max-Cut~\cite{gw95}, coloring 3-colorable graphs~\cite{kms94}, and sparsest cut~\cite{arv09}), quantum complexity theory~\cite{jjuw11}, robust learning and estimation~\cite{cg18,cdg19,cdgw19}, and algorithmic discrepancy and rounding~\cite{bdg16,bg17,b19}. We formally define SDPs with variable size $n \times n$ and $m$ constraints: 
\begin{definition}[Semidefinite programming] \label{defn:sdpprimal}
Given symmetric\footnote{We can assume that $C,A_1,\cdots,A_m$ are symmetric, since given any $M \in \{C,A_1,\cdots,A_m\}$, we have $\sum_{i, j} M_{ij} X_{ij} = \sum_{i, j} M_{ij} X_{ji} = \sum_{i, j} (M^\top)_{ij} X_{ij}$, and therefore we can replace $M$ with $(M+ M^\top)/2$.} matrices $C, A_1, \cdots, A_m \in \mathbb{R}^{n \times n}$ and $b_i \in \R$ for all $i \in [m]$,  
the goal is to solve the convex optimization problem
\begin{align}
\label{eq:sdpprimal}
\max \langle C, X \rangle  \textup{ subject to } X \succeq 0, \langle A_i, X \rangle = b_i \text{ }\forall i \in [m]
\end{align}
where $\langle A,B\rangle := \sum_{i,j} A_{i,j} B_{i,j}$ is the trace product. 
\end{definition} 

\paragraph{Cutting plane and interior point methods}
Two prominent methods for solving SDPs, with runtimes depending {\em logarithmically} on the accuracy parameter $\epsilon$, are the {\em cutting plane method} and the {\em interior point method}. 

The cutting plane method maintains a convex set containing the optimal solution. 
In each iteration, the algorithm queries a separation oracle, which returns a hyperplane that divides the convex set into two subsets. 
The convex set is then updated to contain the subset with the optimal solution.
This process is repeated until the volume of the maintained set becomes small enough and a near-optimal solution can be found. 
Since Khachiyan proved~\cite{k80} that the ellipsoid method solves linear programs in polynomial time, cutting plane methods have played a crucial role in both discrete and continuous optimization~\cite{gls81,gv02}. 

In contrast, interior point methods add a barrier function to the objective and, by adjusting the weight of this barrier function, solve a different optimization problem in each iteration. The solutions to these successive problems form a well-defined {\em central path}. Since 
Karmarkar proved~\cite{k84} that interior point methods can solve linear programs in polynomial time, these methods have become an active research area.
Their number of iterations is usually the square root of the number of dimensions, as opposed to the linear dependence on dimensions in cutting plane methods.

Since cutting plane methods use less structural information than interior point methods, they are  slower at solving almost all  problems where interior point methods are known to apply. However, 
SDPs remain one of the most fundamental optimization problems where the state of the art is, in fact, the opposite: the current fastest cutting plane methods\footnote{\cite{jlsw20} improves upon the runtime of~\cite{lsw15} in terms of the dependence on $\log(n/\epsilon)$, while the polynomial factors are the same in both runtimes.} of~\cite{lsw15,jlsw20} solve a general SDP in time $m(mn^2 +  m^2+ n^\omega)$, while the fastest SDP solvers based on interior point methods in the work of~\cite{nn92} and~\cite{a00} achieve runtimes of $\sqrt{n}(m^2 n^2 + mn^\omega + m^\omega)$ and $(mn)^{1/4}(m^4 n^2 + m^3 n^\omega)$, respectively, which are slower in the most common regime of $m \in [n,n^2]$ (see Table~\ref{table:runtimeComparison}).
This apparent paradox raises the following natural question: 
\begin{center}
{\em How fast can SDPs be solved using interior point methods?}
\end{center}

\subsection{Our results} \label{subsec:results}
We present a faster interior point method for solving SDPs. Our main result is the following theorem, the formal version of which is given in Theorem~\ref{thm:main_formal}. 

\begin{theorem}[Main result, informal]\label{thm:main}
There is an interior point method that solves a general SDP with variable size $n \times n$ and $m$ constraints in time\footnote{We use $O^*$ to hide $n^{o(1)}$ and $\log^{O(1)}(n/\epsilon)$ factors and $\wt{O}$ to hide $\log^{O(1)}(n/\epsilon)$ factors, where $\epsilon$ is the accuracy parameter. } $O^*(\sqrt{n}(mn^2 + m^\omega + n^\omega))$. 
\end{theorem}
\noindent Our runtime can be roughly interpreted as follows:
\begin{itemize}
\item $\sqrt{n}$ is the iteration complexity of the interior point method with the log barrier function. 
\item $mn^2$ is the input size.
\item $m^\omega$ is the cost of inverting the Hessian of the log barrier.
\item $n^\omega$ is the cost of inverting the slack matrix. 
\end{itemize} 
Thus, the terms in the runtime of our algorithm arise as a natural barrier to further speeding up SDP solvers. See  Section~\ref{subsubsec:our_tech}, \ref{subsubsec:bottleneck}, and \ref{subsubsec:fail_attempts} for more detail. 

 Table~\ref{table:SDPhistory} compares our result with previous SDP solvers. The first takeaway of this table and Theorem \ref{thm:main} is that our interior point method always runs faster than that in~\cite{nn92} and is faster than that in~\cite{nn94} and~\cite{a00} when $m \geq n^{1/13}$.  
A second consequence is that whenever $m \geq \sqrt{n}$,  our interior point method is faster than the current fastest cutting plane method~\cite{lsw15,jlsw20}.
We note that $n \leq m \leq n^2$ is satisfied in most SDP applications known to us, such as classical combinatorial optimization problems over graphs, experiment design problems in statistics and machine learning, and sum-of-squares problems. An explicit comparison to previous algorithms in the cases of $m = n$ and $m = n^2$ is shown in Table~\ref{table:runtimeComparison}. 

\begin{table}[!h]
\centering
 \begin{tabular}{ | l | l | l | l | l | }
    \hline
    {\bf Year} & {\bf References} & {\bf Method} & {\bf \#Iters} &  {\bf Cost per iter}   \\ \hline
    1979 & \cite{s77,yn76,k80} & CPM & $m^2$ & $mn^2 + m^2 + n^\omega$ \\ \hline
    1988 & \cite{kte88,nn89} & CPM & $m$ & $mn^2 + m^{3.5}+ n^\omega$ \\ \hline
    1989 & \cite{v89} & CPM & $m$ & $mn^2 + m^\omega + n^\omega $  \\ \hline
    1992 & \cite{nn92} & IPM  & $\sqrt{n}$ & $ m^2 n^2 + m n^\omega + m^\omega $ \\ \hline
    1994 & \cite{nn94, a00} & IPM & $(mn)^{1/4}$ & $m^4 n^2 + m^3 n^\omega$ \\ \hline
    2003 & \cite{km03} & CPM & $m$ & $mn^2 + m^\omega + n^\omega$ \\ \hline
    2015 & \cite{lsw15} & CPM & $m$ & $mn^2 + m^2 + n^\omega $ \\ \hline
    2020 & \cite{jlsw20} & CPM & $m$ & $mn^2 + m^2 + n^\omega$ \\ \hline
    2020 & Our result & IPM & $\sqrt{n}$ & $mn^2 + m^\omega + n^{\omega}$\\ \hline
  \end{tabular}
  \caption{Summary of key SDP algorithms. CPM stands for cutting plane method, and IPM, interior point method. $n$ is the size of the variable matrix, and $m \leq n^2$ is the number of constraints. Runtimes hide $n^{o(1)}$, $m^{o(1)}$ and $\poly\log(1/\epsilon)$ factors, where $\epsilon$ is the accuracy parameter. \cite{a00} simplifies the proofs in~\cite[Section 5.5]{nn94}. Neither~\cite{a00} nor~\cite{nn94} explicitly analyzed their runtimes, and their runtimes shown here are our best estimates.
} \label{table:SDPhistory}
\end{table}

\begin{table*}[!h]
\centering
 \begin{tabular}{ | l | l | l | l | l | l | l | }
    \hline
    \multirow{2}{*}{\bf Year} & \multirow{2}{*}{\bf References} & \multirow{2}{*}{\bf Method} &  \multicolumn{2}{|c|}{\bf Runtime} \\ 
    \cline{4-5}
    & & &  $m=n$ & $m=n^2$ \\ \hline
    1979 & \cite{s77,yn76,k80} & CPM & $n^5$ & $n^8$ \\ \hline
    1988 & \cite{kte88,nn89} & CPM & $n^{4.5}$ & $n^9$ \\ \hline
    1989 & \cite{v89} & CPM & $n^4$ & $n^{6.746}$ \\ \hline
    1992 & \cite{nn92} & IPM  & $n^{4.5}$ & $n^{6.5} $ \\ \hline
    1994 & \cite{nn94, a00} & IPM & $n^{6.5}$ & $n^{10.75}$ \\ \hline
    2003 & \cite{km03} & CPM & $n^4$ & $n^{6.746}$ \\ \hline
    2015 & \cite{lsw15} & CPM & $n^4$ & $n^6$ \\ \hline
    2020 & \cite{jlsw20} & CPM & $n^4$ & $n^6$ \\ \hline
    2020 & Our result & IPM & $n^{3.5}$ & $n^{5.246}$ \\ \hline
  \end{tabular}
  \caption{Total runtimes for the algorithms in Table~\ref{table:SDPhistory} for SDPs when $m=n$ and $m=n^2$, where $n$ is the size of matrices, and $m$ is the number of constraints. The runtimes shown in the table hide $n^{o(1)}$, $m^{o(1)}$ and $\poly\log(1/\epsilon)$ factors, where $\epsilon$ is the accuracy parameter and assume $\omega$  to equal its currently best known upper bound of $2.373$. } \label{table:runtimeComparison}
\end{table*}

Even in the more general case where the SDP might not be dense, where $\nnz(A)$ is the input size (i.e., the total number of non-zeroes in all matrices $A_i$ for $i \in [m]$ and $C$), our interior point method runs faster than the current fastest cutting plane methods\cite{lsw15,jlsw20}, which run in time $O^*(m(\nnz(A)  + m^2+ n^\omega))$.
\begin{restatable}[Comparison with Cutting Plane Method]{theorem}{CompareCuttingPlane} \label{thm:compare_cutting_plane} 
When $m \geq n$, there is an interior point method that solves an SDP with $n \times n$ matrices, $m$ constraints, and $\nnz(A)$ input size, faster than the current best cutting plane method~\cite{lsw15,jlsw20}, over all regimes of $\nnz(A)$.  
\end{restatable} 
   %%% Section 1. Introduction

\subsection{Technique overview}

\subsubsection{Interior point method for solving SDPs}\label{subsubsec:ipmintro}

By removing redundant constraints, we can, without loss of generality, assume $m \leq n^2$ in the primal formulation of the SDP~\eqref{eq:sdpprimal}. Thereafter, instead of solving the primal SDP, which has variable size $n \times n$, we solve its dual formulation, which has dimension $m \leq n^2$:
\begin{equation}
\begin{aligned}
\label{eq:sdpdual}
\min b^\top y  \textup{ subject to } S = \sum_{i=1}^m y_i A_i - C,\textup{ and } S \succeq 0. 
\end{aligned}
\end{equation}

Interior point methods solve \eqref{eq:sdpdual} by minimizing the penalized objective function: 
\begin{align} \label{eq:pathfollowing}
\min_{y \in \mathbb{R}^m} f_\eta(y), \textup{ where }  f_\eta(y)  := \eta \cdot b^\top y + \phi(y) ,
\end{align}
where $\eta > 0$ is a parameter and $\phi: \mathbb{R}^m \rightarrow \mathbb{R}$ is a {\em barrier} function that approaches infinity as $y$ approaches the boundary of the feasible set $\{y \in \mathbb{R}^m: \sum_{i=1}^m y_i A_i \succeq C\}$. These methods first obtain an approximate minimizer of $f_\eta$ for some small $\eta>0$, which they then use as an initial point to minimize $f_{(1+c)\eta}$, for some constant $c > 0$, via the Newton method. 
This process repeats until the parameter $\eta$ in~\eqref{eq:pathfollowing} becomes sufficiently large, at which point the minimizer of $f_\eta$ is provably close to the optimal solution of~\eqref{eq:sdpdual}. The iterates $y$ generated by this method follow a \textit{central path}. Different choices of the barrier function $\phi$ lead to different run times in solving ~\eqref{eq:pathfollowing}, as we next describe. 

\paragraph{The log barrier} Nesterov and Nemirovski~\cite{nn92} use the {\em log barrier} function, 
\begin{align} \label{eq:logbarrier}
\phi(y) = g(y) := - \log \det \left(\sum_{i=1}^m y_i A_i - C \right),
\end{align}   in ~\eqref{eq:pathfollowing} and, in $O(\sqrt{n} \log(n/\epsilon))$ iterations, obtain a feasible dual solution $y$ that satisfies $b^\top y \leq b^\top y^* + \epsilon$, where $y^* \in \R^m$ is the optimal solution for ~\eqref{eq:sdpdual}. Within each iteration, the costliest step is to compute the inverse of the Hessian of the log barrier function for the Newton step. For each $(j,k) \in [m] \times [m]$, the $(j,k)$-th entry of $H$ is given by 
\begin{align} \label{eq:HessianLogBarrier}
H_{j,k} = \tr[S^{-1} A_j S^{-1} A_k]. 
\end{align} 
The analysis of~\cite{nn92} first computes $S^{-1/2} A_j S^{-1/2}$ for all $j \in [m]$, which takes time $O^*(m n^\omega)$, and then calculates the $m^2$ trace products $\tr[S^{-1} A_j S^{-1} A_k]$ for all $(j,k) \in [m] \times [m]$, each of which takes $O(n^2)$ time. 
Inverting the Hessian costs $O^*(m^\omega)$, which results in a total runtime of $O^*(\sqrt{n}(m^2 n^2 + m n^\omega + m^\omega ))$. 
%Leveraging fast matrix multiplication, it's not hard to improve this runtime to $O^*(m n^\omega + m^{\omega-2} n^2)$ per iteration. In fact, computing $S^{-1/2} A_j S^{1/2}$ for all $j \in [m]$ takes time $O^*(m n^\omega)$; after flattening $S^{-1/2} A_j S^{-1/2}$ into a vector of length $n^2$ and stacking all $m$ such vectors for $j \in [m]$ into a matrix $B$ of size $m \times n^2$, the Hessian $H = B B^\top$ can be computed in time $O^*(m^{\omega-1} n^2)$ by cutting $B$ into square blocks of size $m \times m$. The overall runtime of this approach is thus $O^*(\sqrt{n}(mn^\omega + m^{\omega-1} n^2))$. 

\paragraph{The volumetric barrier} Vaidya~\cite{v89} introduced the {\em volumetric barrier} for a polyhedral set $\{x \in \mathbb{R}^n: A x \geq c\}$, where $A \in \mathbb{R}^{m \times n}$ and $c \in \mathbb{R}^m$. Nesterov and Nemirovski~\cite{nn94} studied the following extension of the volumetric barrier to the convex subset $\{y \in \mathbb{R}^m: \sum_{i=1}^m y_i A_i \succeq C\}$ of the polyhedral cone:
\begin{align*}
V(y) = \frac{1}{2} \log \det(\nabla^2 g(y)) ,
\end{align*}
where $g(y)$ is the log barrier function defined in \eqref{eq:logbarrier}. They proved that choosing $\phi(y) = \sqrt{n} V(y)$ in~\eqref{eq:pathfollowing} makes the interior point method converge in $\widetilde{O}(\sqrt{m} n^{1/4})$ iterations, which is smaller than the $\widetilde{O}(\sqrt{n})$ iteration complexity of~\cite{nn92} when $m \leq \sqrt{n}$.  
They also studied the {\em combined volumetric-logarithmic barrier}
\begin{align*}
V_\rho(y) = V(y) + \rho \cdot g(y) 
\end{align*}
and showed that taking $\phi(y) = \sqrt{n/m} \cdot V_\rho(y)$ for $\rho = (m-1)/(n-1)$ yields an iteration complexity of $\widetilde{O}((mn)^{1/4})$. when $m \leq n$, this iteration complexity is lower than $\widetilde{O}(\sqrt{n})$ of~\cite{nn92}. We refer readers to the much simpler proofs in~\cite{a00} for these results. 

However, the volumetric barrier (and thus the combined volumetric-logarithmic barrier) leads to complicated expressions for the gradient and Hessian that make each iteration costly. 
For instance, the Hessian of the volumetric barrier is 
\begin{align*}
\nabla^2 V(y) = 2 Q(y) + R(y) - 2 T(y) ,
\end{align*}
where $Q(y)$, $R(y)$, and $T(y)$ are $m \times m$ matrices such that for each $(j,k) \in [m] \times [m]$,
\ifdefined\isfocs
\begin{align} \label{eq:volumetricbarrier_Hessian}
Q(y)_{j,k} & = \tr\left[\mathcal{A} H^{-1} \mathcal{A}^\top \left( \left( B_j S B_k \right) \wh{\otimes} S^{-1} \right) \right], \nonumber\\
R(y)_{j,k} & = \tr \left[\mathcal{A} H^{-1} \mathcal{A}^\top \left( B_j \wh{\otimes} B_k \right) \right], \\
T(y)_{j,k} & = \tr\left[ \mathcal{A} H^{-1} \mathcal{A}^\top ( B_j \wh{\otimes} S^{-1} )  \mathcal{A} H^{-1} \mathcal{A}^\top ( B_k \wh{\otimes} S^{-1} ) \right] . \nonumber 
\end{align} 
where $B_j :=  S^{-1} A_j S^{-1} $.

\else
\begin{align} \label{eq:volumetricbarrier_Hessian}
Q(y)_{j,k} & = \tr\left[\mathcal{A} H^{-1} \mathcal{A}^\top \left( \left(S^{-1} A_j S^{-1} A_k S^{-1} \right) \wh{\otimes} S^{-1} \right) \right], \nonumber\\
R(y)_{j,k} & = \tr \left[\mathcal{A} H^{-1} \mathcal{A}^\top \left( \left(S^{-1} A_j S^{-1} \right) \wh{\otimes} \left( S^{-1} A_k S^{-1}  \right) \right) \right], \\
T(y)_{j,k} & = \tr\left[ \mathcal{A} H^{-1} \mathcal{A}^\top \left( \left( S^{-1} A_j S^{-1} \right) \wh{\otimes} S^{-1} \right)  \mathcal{A} H^{-1} \mathcal{A}^\top \left( \left( S^{-1} A_k S^{-1} \right) \wh{\otimes} S^{-1} \right) \right] . \nonumber 
\end{align}  
\fi
Here, $\mathcal{A} \in \mathbb{R}^{n^2 \times m}$ is the $n^2 \times m$ matrix whose $i$th column is obtained by flattening $A_i$ into a vector of length $n^2$, and $\wh{\otimes}$ is the symmetric Kronecker product
\begin{align*}
A \wh{\otimes} B := \frac{1}{2} (A \otimes B + B \otimes A) ,
\end{align*}
where $\otimes$ is the Kronecker product (see \ifdefined\isfocs full version \else Section~\ref{ssec:notation}\fi for formal definition). 
Due to the complicated formulas in \eqref{eq:volumetricbarrier_Hessian}, efficient computation of Newton step in each iteration of the interior point method is difficult; in fact, each iteration runs slower than the Nesterov-Nemirovski interior point method by a factor of $m^2$. Since most  applications of SDPs known to us have the number of constraints $m$ be at least linear in $n$,
the total runtime of interior point methods based on the volumetric barrier and the combined volumetric-logarithmic barrier is inevitably slow.

%%%%%%%%%%%%%%%%%%%%%%%%%%%%%%%%%%%%%
%%%%%%%%%%%%%%%%%%%%%%%%%%%%%%%%%%%%%
%%%%%%%%%%%%%%%%%%%%%%%%%%%%%%%%%%%%%

\subsubsection{Our techniques}  \label{subsubsec:our_tech}

Given the inefficiency of implementing the volumetric and volumetric-logarithmic barriers discussed above, this paper uses the log barrier in \eqref{eq:logbarrier}. 
We now describe some of our key techniques that improve the runtime of the Nesterov-Nemirovski interior point method~\cite{nn92}. 

\medskip

\paragraph{Hessian computation using fast rectangular matrix multiplication} As noted in Section~\ref{subsubsec:ipmintro}, the runtime bottleneck in~\cite{nn92} is computing the inverse of the Hessian  of the log barrier function, where the Hessian is described in \eqref{eq:HessianLogBarrier}. In~\cite{nn92}, each of these $m^2$ entries is computed separately, resulting in a runtime of $O(m^2 n^2)$ per iteration.

Instead contrast, we show below how to group these computations using rectangular matrix multiplication. 
The expression from \eqref{eq:HessianLogBarrier} can be re-written as 
\begin{align} \label{eq:Hessiantrue}
H_{j,k} = \tr[S^{-1/2} A_j S^{-1/2} \cdot S^{-1/2} A_k S^{-1/2}].
\end{align}
We first compute the key quantity $S^{-1/2} A_j S^{-1/2} \in \R^{n \times n}$ for all $j \in [m]$ by stacking all matrices $A_j \in \R^{n \times n}$ into a tall matrix of size $mn \times n$, and then compute the product of $S^{-1/2} \in \R^{n \times n}$ with this tall matrix. 
This matrix product can be computed in time $\Tmat(n,mn,n)$\footnote{\ifdefined\isfocs We define $\Tmat(n,r,m)$ to be the number of operations needed to compute the product of matrices of dimensions $n \times r$ and $r \times m$.  \else See Section~\ref{subsec:matproperty} for the definition. \fi } using fast rectangular matrix multiplication. We then flatten each $S^{-1/2} A_j S^{-1/2}$ into a row vector of length $n^2$ and stack all $m$ vectors to form a matrix $\mathcal{B}$ of size $m \times n^2$, i.e., the $j$-th row of $\mathcal{B}$ is $\mathcal{B}_{j} = \vect(S^{-1/2} A_j S^{-1/2})$. It follows that the Hessian can be computed as
\begin{align} \label{eq:Hessianrect}
H = \mathcal{B}\mathcal{B}^\top,
\end{align}
which takes time $\Tmat(m,n^2,m)$ by applying fast rectangular matrix multiplication. 
By leveraging recent developments in this area~\cite{gu18}, this approach already improves upon the runtime in~\cite{nn92}. 

Thus far, we have reduced the per iteration cost of $O^*(m^2 n^2 + m n^\omega)$ for Hessian computation down to
\begin{align*}
\Tmat( n, mn, n ) + \Tmat( m, n^2, m ) .
\end{align*}

\medskip

\paragraph{Low rank update on the slack matrix}
The fast rectangular matrix multiplication approach noted above, however, is still not very efficient, because the Hessian must be computed from scratch in each iteration of the interior point method. If there are $T$ iterations in total, it then takes time
\begin{align*}
T \cdot ( \Tmat(n,mn,n) + \Tmat(m,n^2,m) ) .
\end{align*}
To further improve the runtime, we need to efficiently update the Hessian for the current iteration from the Hessian computed in the previous one. 
Generally, this is not possible, as the slack matrix $S \in \R^{n \times n}$ in \eqref{eq:Hessiantrue} might change arbitrarily in the Nesterov-Nemirovski interior point method. 

To overcome this problem, we propose a new interior point method that maintains an approximate slack matrix $\widetilde{S} \in \R^{n \times n}$, which is a spectral approximation of the true slack matrix $S \in \R^{n \times n}$ such that $\widetilde{S}$ admits a {\em low-rank update} in each iteration. Where needed, we will now use the subscript $t$ to denote a matrix in the $t$-th iteration. 
Our algorithm updates only the directions in which $\widetilde{S}_t$ deviates too much from $S_{t+1}$; the changes to $S_t$ for the remaining directions are not propagated in $\widetilde{S}_t$. 
This process of selective update ensures a low-rank change in $\widetilde{S}_t$ even when $S_t$ suffers from a full-rank update; it also guarantees the proximity of the algorithm's iterates to the central path. 
Specifically, for each iteration $t \in [T]$, we define the {\em difference matrix} 
\begin{align*}
Z_t = S_t^{-1/2} \widetilde{S}_t S_t^{-1/2} - I ~~~ \in \R^{n \times n},
\end{align*}
which intuitively captures how far the approximate slack matrix $\widetilde{S}_t$ is from the true slack matrix $S_t$. 
We maintain the invariant $\| Z_t \|_{\op} \leq c$ for some sufficiently small constant $c > 0$.  
In the $(t+1)$-th iteration when $S_t$ gets updated to $S_{t+1}$, our construction of $\widetilde{S}_{t+1}$ involves a novel approach of zeroing out some of the largest eigenvalues of $|Z_t|$ to bound the rank of the update  on the approximate slack matrix.

We prove that with this approach, the updates on $\widetilde{S} \in \R^{n \times n}$ over all $T = \widetilde{O}(\sqrt{n})$ iterations satisfy the following {\em rank inequality} (see \ifdefined\isfocs full version \else Theorem~\ref{thm:rankineq}\fi for the formal statement).
\begin{theorem}[Rank inequality, informal version] \label{thm:rankineq_intro} 
Let $\wt{S}_1, \wt{S}_2, \cdots, \wt{S}_T \in \R^{n \times n}$ denote the sequence of approximate slack matrices generated in our interior point method. For each $t \in [T-1]$, denote by $r_t = \rank(\wt{S}_{t+1} - \wt{S}_t )$  the rank of the update on $\wt{S}_t$.
Then, the  sequence $r_1, r_2, \cdots, r_T$ satisfies 
\begin{align*}
\sum_{t=1}^T \sqrt{r_t} = \widetilde{O}(T) . 
\end{align*} 
\end{theorem}
\noindent The key component to proving Theorem~\ref{thm:rankineq_intro} is the  potential function $\Phi : \R^{n \times n} \rightarrow \R_{\geq 0}$
\begin{align*}
\Phi(Z) := \sum_{ \ell = 1 }^n \frac{|\lambda(Z)|_{[\ell]}}{\sqrt{\ell}} ,
\end{align*}
where $|\lambda(Z)|_{[\ell]}$ is the $\ell$-th in the list of eigenvalues of $Z \in \R^{n \times n}$ sorted in decreasing order of their absolute values. We show an upper bound on the increase in this potential when $S$ is updated, a lower bound on its decrease when $\widetilde{S}$ is updated, and combine the two with non-negativity of the potential to obtain Theorem~\ref{thm:rankineq_intro}. 

Specifically, first we prove that whenever $S$ is updated in an iteration, the potential function {\em increases} by at most $\widetilde{O}(1)$ (see \ifdefined\isfocs full version \else Lemma~\ref{lem:S_change_pot}\fi). 
The proof of this statement crucially uses the structural property of interior point method that slack matrices in consecutive steps are sufficiently close to each other. 
Formally, for any iteration $t \in [T]$, we show in \ifdefined\isfocs the full version \else Theorem~\ref{thm:approx-central-path} \fi that the consecutive slack matrices $S_t$ and $S_{t+1}$ satisfy  
\begin{align}\label{eq:S_change_slowly} 
\| S_t^{-1/2} S_{t+1} S_t^{-1/2} - I \|_F = O(1) 
\end{align}
and combine this bound with the Hoffman-Wielandt theorem \cite{hj12}, which relates the $\ell_2$ distance between the spectrum of two matrices with the Frobenius norm of their difference (see \ifdefined\isfocs full version\else Fact~\ref{fac:hw_thm}\fi). Next, when $\widetilde{S}$ gets updated, we prove that our method of zeroing out the $r_t$ largest eigenvalues of $|Z_t|$, thereby incurring a rank-$r_t$ update to $\widetilde{S}_t$, results in a potential decrease of at least $\widetilde{O}(\sqrt{r_t})$ (see \ifdefined\isfocs full version\else Lemma~\ref{lem:tildeS_change_pot}\fi).

\medskip 

\paragraph{Maintaining rectangular matrix multiplication for Hessian computation\ifdefined\isfocs\else.\fi}

Given the low-rank update on $\wt{S}$ described above, we show how to efficiently update the {\em approximate Hessian} $\widetilde{H}$, defined as 
\begin{align} \label{eq:approxHessiantrue}
\widetilde{H}_{j,k} = \tr[\wt{S}^{-1} A_j \wt{S}^{-1} A_k]
\end{align} 
for each entry $(j,k) \in [m] \times [m]$. The approximate slack matrix $\wt{S}$ being a spectral approximation of the true slack matrix $S$ implies that  the approximate Hessian $\wt{H}$ is also a spectral approximation of the true Hessian $H$ (see \ifdefined\isfocs full version\else Lemma~\ref{lem:appslackappH}\fi). This approximate Hessian therefore suffices  for our algorithm to approximately follow the central path. 

To efficiently update the approximate Hessian $\wt{H}$ in \eqref{eq:approxHessiantrue}, 
we notice that a rank-$r$ update on $\wt{S}$ implies a rank-$r$ update on $\wt{S}^{-1}$ via the Woodbury matrix identity (see \ifdefined\isfocs full version\else Fact~\ref{fac:woodbury}\fi). 
The change in $\wt{S}^{-1}$ can be expressed as  
\begin{align} \label{eq:DeltaStildeInv}
\Delta (\wt{S}^{-1}) = V_+ V_+^\top - V_- V_-^\top, 
\end{align}
where $V_+, V_- \in \mathbb{R}^{n \times r}$. 
Plugging \eqref{eq:DeltaStildeInv} into \eqref{eq:approxHessiantrue}, we can express $\Delta \wt{H}_{j,k}$ as the sum of multiple terms, among the costliest of which are those of the form $\tr[\wt{S}^{-1} A_j V V^\top A_k]$, where $V \in \mathbb{R}^{n \times r}$ is either $V_+$ or $V_-$.  
We compute $\tr[\wt{S}^{-1} A_j V V^\top A_k]$ for all $(j,k) \in [m] \times [m]$ in time $\Tmat(r, n, mn)$ by first computing $V^\top A_k$ for all $k \in [m]$ by horizontally concatenating all $A_k$'s into a wide matrix of size $n \times mn$. We then compute the product of $\wt{S}^{-1/2}$ with $A_j V$ for all $j \in [m]$, which can be done in time $\Tmat(n, n, m r)$, which equals $\Tmat(n, m r, n)$ (see \ifdefined\isfocs full version\else Lemma~\ref{lem:OrderOfTmat}\fi). 
Finally, by flattening each $\wt{S}^{-1/2} A_j V$ into a vector of length $n r$ and stacking all these vectors to form a matrix $\wt{\mathcal{B}} \in \mathbb{R}^{m \times n r}$ with $j$-th row  
\begin{align*}
\wt{\mathcal{B}}_j = \vect(\wt{S}^{-1/2} A_j V) ,
\end{align*}
the task of computing $\tr[\wt{S}^{-1} A_j V V^\top A_k]$ for all $(j,k) \in [m] \times [m]$ reduces to computing $\wt{\mathcal{B}}  \wt{\mathcal{B}}^\top$, which costs $\Tmat(m, n r, m)$. 

In this way, we reduce the runtime of $T \cdot (\Tmat(n,mn,n) + \Tmat(m,n^2,m) )$ for computing the Hessian using fast rectangular matrix multiplication down to
\begin{align} \label{eq:faster_runtime}
\sum_{t=1}^T  ~\left( \Tmat(r_t, n, mn) + \Tmat(n, m r_t,n)+ \Tmat(m, n r_t, m) \right) ,
\end{align}
where $r_t$ is the rank of the update on $\wt{S}_t$. 
Applying Theorem~\ref{thm:rankineq_intro} with several properties of fast rectangular matrix multiplication that we prove in \ifdefined\isfocs the full version \else Section~\ref{subsec:matproperty} \fi, we upper bound the runtime in \eqref{eq:faster_runtime} by 
\begin{align*}
O^*(\sqrt{n}(mn^2 + m^\omega + n^\omega)),
\end{align*}  
which implies Theorem~\ref{thm:main}. In Section~\ref{subsubsec:bottleneck} and~\ref{subsubsec:fail_attempts}, we discuss bottlenecks to further improving our runtime.

%%%%%%%%%%%%%%%%%%%%%%%%%%
%%%%%%%%%%%%%%%%%%%%%%%%%%
%%%%%%%%%%%%%%%%%%%%%%%%%%

\subsubsection{Bottlenecks of our interior point method} \label{subsubsec:bottleneck}

%Several barriers prevent our improving our  runtime of $\sqrt{n}(mn^2 + m^\omega + n^\omega)$. 
%for further improving the runtime of $\sqrt{n}(mn^2 + m^\omega + n^\omega)$ of our interior point method. 
%Notice that in our per iteration cost, the $n^\omega$ term comes from inverting the slack matrix $S$ and the $m^\omega$ arises as a result of inverting the (approximate) Hessian.
%We believe that both these operations are in fact necessary for interior point methods, and thus it is less likely that one could improve the terms $m^\omega$ and $n^\omega$. 
In most cases, the costliest term in our runtime is the per iteration cost of $mn^2$, which corresponds to reading the entire input in each iteration. 
Our subsequent discussions therefore focus on the steps in our algorithm that require at least $mn^2$ time per iteration. 

%bottlenecks for improving the first term $mn^2$, i.e. the cost of reading the entire input in each interation.

\medskip

\paragraph{Slack matrix computation\ifdefined\isfocs\else.\fi}

When $y$ is updated in each iteration of our interior point method, we need to compute the true slack matrix $S$ as 
\begin{align*}
S = \sum_{i \in [m]} y_i A_i - C.
\end{align*}
Computing $S$ is needed to update the approximate slack matrix $\wt{S}$ so that $\wt{S}$ remains a spectral approximation to $S$.
As $S$ might suffer from full-rank changes, it naturally requires $m n^2$ time to compute in each iteration. 
This is the first appearance of the $mn^2$ cost per iteration. 

\medskip

\paragraph{Gradient computation\ifdefined\isfocs\else.\fi}

Recall from \eqref{eq:pathfollowing} that our interior point method follows the central path defined via the penalized objective function
\begin{align*}
\min_{y \in \mathbb{R}^m} f_\eta(y) \quad \text{where} \quad f_\eta(y)  := \eta b^\top y + \phi(y),
\end{align*}
for a parameter $\eta > 0$ and $\phi(y) = - \log \det S$.
In each iteration, to perform the Newton step, the gradient of the penalized objective is computed as 
\begin{align} \label{eq:gradientcomputation}
g_\eta(y)_j = \eta \cdot b_j - \tr[S^{-1} A_j]
\end{align} 
for each coordinate $j \in [m]$. Even if we are given $S^{-1}$, it still requires $mn^2$ time to compute \eqref{eq:gradientcomputation} for all $j \in [m]$. 
This is the second appearance of the per iteration cost of $mn^2$.

\medskip

\paragraph{Approximate Hessian computation\ifdefined\isfocs\else.\fi}

Recall from Section~\ref{subsubsec:our_tech} that updating the approximate slack matrix $S$  by rank $r$ means the time needed to update the approximate Hessian is dominated by computing the term
\begin{align*}
 \Delta_{j,k} = \tr[\wt{S}^{-1/2} A_j V \cdot V^\top A_k \wt{S}^{-1/2}] ,
\end{align*}
where $V \in \mathbb{R}^{n \times r}$ is a tall, skinny matrix that comes from the spectral decomposition of $\Delta \wt{S}^{-1}$. 
Computing $\Delta_{j,k}$ for all $(j,k) \in [m] \times [m]$ requires reading at least $A_j$ for all $j \in [m]$, which takes time $mn^2$. 
This is the third bottleneck that leads to  the $mn^2$ term in the cost per iteration.

%%%%%%%%%%%%%%%%%%%%%%%%%%%%%%%%
%%%%%%%%%%%%%%%%%%%%%%%%%%%%%%%%
%%%%%%%%%%%%%%%%%%%%%%%%%%%%%%%%

\subsubsection{LP techniques unlikely to improve SDP runtime\ifdefined\isfocs\else\fi} \label{subsubsec:fail_attempts}

The preceeding discussion of bottlenecks suggests that reading the entire input in each iteration, which takes $mn^2$ time per iteration, stands as a natural barrier to further improving the runtime of SDP solvers based on interior point methods. 

In the context of linear programming (LP), several recent results~\cite{cls19,blss20} yield faster interior point methods that bypass reading the entire input in 
every iteration. 
Two techniques crucial to these results are: (1) showing that the Hessian (projection matrix) admits low-rank updates, and (2) speeding computation of the Hessian via sampling. 

We now describe these techniques in the context of SDP and argue that they are unlikely to improve our runtime. 

\medskip

\paragraph{Showing that the Hessian admits low-rank updates\ifdefined\isfocs\else.\fi}

We saw in Section~\ref{subsubsec:our_tech} that constructing an approximate slack matrix $\wt{S}$ that admits low-rank updates in each iterations leveraged the fact that the true slack matrix $S$ changes ``slowly'' throughout our interior point method as described in~\eqref{eq:S_change_slowly}. One natural question that follows is whether a similar upper bound can be obtained for the Hessian. 
If such a result could be proved, then one could maintain an approximate Hessian that admitted low-rank updates, which would speed up the approximate Hessian computation. 
Indeed, in the context of LP, such a bound for the Hessian  can be proved (e.g., ~\cite[Lemma 47]{blss20}). 

Unfortunately, it is impossible to prove such a statement for the Hessian in the context of SDP. 
To show this, it is convenient to express the Hessian using the Kronecker product \ifdefined\isfocs full version \else (Section~\ref{ssec:notation})\fi as
\begin{align*}
H = \mathcal{A}^\top \cdot (S^{-1} \otimes S^{-1}) \cdot \mathcal{A} ,
\end{align*}
where $\mathcal{A} \in \mathbb{R}^{n^2 \times m}$ is the $n^2 \times m$ matrix whose $i$th column is obtained by flattening $A_i$ into a vector of length $n^2$. 
By proper scaling, we can assume without loss of generality that the current slack matrix is $S = I$, and the slack matrix in the next iteration is $S_{\mathrm{new}} = I + \Delta S$, which satisfies $\norm{\Delta S}_F = c$ for some tiny constant $c > 0$.
Consider the simple example where $\mathcal{A} = I$ (we are assuming here that $m=n^2$ so that $\mathcal{A}$ is a square matrix), which implies that the change in the Hessian can be approximately computed as
\ifdefined\isfocs
\begin{align*}
 & ~ \norm{H^{-1/2} \Delta H H^{-1/2}}_F^2 \\
& \approx \tr\left[ \left( \left(I - \Delta S \right) \otimes \left(I - \Delta S \right) - I \otimes I \right)^2 \right] \\
& \approx \tr\left[ (I \otimes \Delta S + \Delta S \otimes I)^2  \right] \\
& \geq 2 \cdot \tr[I^2] \cdot \tr \left [ (\Delta S)^2 \right] \\
& = 2 n \norm{\Delta S}_F^2 ~ \gg ~ 1 .
\end{align*}
\else
\begin{align*}
\norm{H^{-1/2} \Delta H H^{-1/2}}_F^2 
& \approx \tr\left[ \left( \left(I - \Delta S \right) \otimes \left(I - \Delta S \right) - I \otimes I \right)^2 \right] \\
& \approx \tr\left[ (I \otimes \Delta S + \Delta S \otimes I)^2  \right] \\
& \geq 2 \cdot \tr[I^2] \cdot \tr \left [ (\Delta S)^2 \right] \\
& = 2 n \norm{\Delta S}_F^2 ~ \gg ~ 1 .
\end{align*}
\fi
This large change indicates that we are unlikely to obtain an approximation to the Hessian that admits low-rank updates, which is a key difference between LP and SDP.  

\medskip

\paragraph{Sampling for faster Hessian computation\ifdefined\isfocs\else.\fi}

Recall from \eqref{eq:Hessianrect} that the Hessian can be computed as
\begin{align*}
H = \mathcal{B} \cdot \mathcal{B}^\top, 
\end{align*}
where the $j$th row of $\mathcal{B} \in \mathbb{R}^{m \times n^2}$ is $\mathcal{B}_j = \vect(S^{-1/2} A_j S^{-1/2})$ for all $j \in [m]$. 
We might attempt to approximately compute $H$ faster by sampling a subset of columns of $\mathcal{B}$ indexed by $L \subseteq [n^2]$ and compute the product for only the sampled columns. 
This could reduce the dimension of the matrix multiplication and speed up the Hessian computation. 
Indeed, sampling techniques have been successfully used to obtain faster LP solvers~\cite{cls19,blss20}.

For SDP, however, sampling is unlikely to speed up the Hessian computation. 
In general, we must sample at least $m$ columns (i.e. $|L| \geq m$) of $\mathcal{B}$ to spectrally approximate $H$ or the computed matrix will not be full rank.  
However, this requires computing the entries of $S^{-1/2} A_j S^{-1/2}$ that correspond to $L \subseteq [n^2]$ for all $j \in [m]$, which requires reading all $A_j$'s and thus still takes $O(mn^2)$ time.

\subsection{Related work}

\paragraph{Linear Programming.} Linear Programming is a class of fundamental problems in convex optimization. There is a long list of work focused on fast algorithms for linear programming \cite{d47,k80,k84,v87,v89_lp,ls14,ls15,sidford15,lee16,cls19,b20,blss20}.

\paragraph{Cutting Plane Method.} Cutting plane method is a class of optimization methods that iteratively refine a convex set that contains the optimal solution by querying a separation oracle. Since its introduction in the 1950s, there has been a long line of work on obtaining fast cutting plane methods \cite{s77,yn76,k80,kte88,nn89,v89,av95,bv02,lsw15,jlsw20}.

\paragraph{First-Order SDP Algorithms.} 
As the focus of this paper, cutting plane methods and interior point methods solve SDPs in time that depends {\em logarithmically} on $1/\epsilon$, where $\epsilon$ is the accuracy parameter. 
A third class of algorithms, the {\em first-order methods}, solve SDPs at runtimes that depend {\em polynomially} on $1/\epsilon$. While having worse dependence on $1/\epsilon$ compared to IPM and CPM, these first-order algorithms usually have better dependence on the dimension. There is a long list of work on first-order methods for general SDP or special classes of SDP (e.g. Max-Cut SDP \cite{ak07, gh16, al17, cdst19, lp19, ytfuc19}, positive SDPs \cite{jy11,pt12,alo16,jllpt20}.)

 %%% Section 1.3, Related work

%The first ten pages stop here.

\ifdefined\isfocs
 
\else
\section{Preliminaries}\label{sec:prelims}

\subsection{Notation}
\label{ssec:notation} 
For any integer $d$, we use $[d]$ to denote the set $\{1,2,\cdots,d\}$. We use $\mathbb{S}^{n \times n}$ to denote the set of symmetric $n \times n$ matrices, $\mathbb{S}^{n \times n}_{\geq 0}$ for the set of $n \times n$ positive semidefinite matrices, and $\mathbb{S}^{n \times n}_{>0}$ for the set of $n\times n$ positive definite matrices. For two matrices $A,B \in \mathbb{S}^{n \times n}$, the notation $A \preceq B$ means that $B - A \in \mathbb{S}^{n \times n}_{\geq 0}$. When clear from the context, we use $0$ to denote the all-zeroes matrix (e.g. $A \succeq 0$). For a vector $v \in \R^n$, we use $\diag(v)$ to denote the diagonal $n \times n$ matrix with $\diag(v)_{i,i} = v_i$. For $A, B \in \mathbb{S}^{n \times n}$, we define the inner product to be the trace product of $A$ and $B$, defined as $\langle A, B \rangle := \tr[A^\top B] = \sum_{i, j \in [n]} A_{i,j} B_{i,j}$. For two matrices $A \in \mathbb{R}^{m \times n}$ and $B \in \mathbb{R}^{k \times l}$, the {\em Kronecker product} of $A$ and $B$, denoted as $A \otimes B$, is defined as the $mk \times nl$ block matrix whose $(i,j)$ block is $A_{i,j} B$, for all $(i,j) \in [m] \times [n]$.  

Throughout this paper, unless otherwise specified, $m$ denotes the number of constraints for the primal SDP \eqref{eq:sdpprimal}, and the variable matrix $X$ is of size $n\times n$. The number of non-zero entries in all the $A_i$ and $C$  of \eqref{eq:sdpprimal} is denoted by $\nnz(A)$. 

\subsection{Useful facts}
\label{ssec:facts}
\paragraph{Linear algebra.} Some matrix norms we frequently use in this paper are the Frobenius and operator norms, defined as follows. The Frobenius norm of a matrix $A \in \R^{n\times n}$ is defined to be $\| A \|_F := \sqrt{\tr [ A^\top A ]}$. The operator (or spectral) norm $\|A\|_{\mathrm{op}}$ of $A \in \R^{n \times n}$ is defined to be the largest singular value of $A$. In the case of symmetric matrices (which is what we encounter in this paper), this can be shown to equal the largest absolute eigenvalue of the matrix. A property of trace we frequently use is the following: given matrices $A_1 \in \R^{m \times n_1}, A_2 \in \R^{n_1 \times n_2}, \dotsc, A_k \in \R^{n_{k-1} \times n_k}$, the trace of their product is invariant under cyclic permutation $\tr [ A_1 A_2 \dotsc A_k ] = \tr [ A_2 A_3 \dotsc A_k A_1 ] = \cdots = \tr [ A_k A_1 \dotsc A_{k-2} A_{k-1} ] $. 
A matrix $A \in \mathbb{R}^{n \times n}$ is called {\em normal} if $A$ commutes with its transpose, i.e. $A A^\top = A^\top A$. 
We note that all symmetric $n \times n$ matrices are normal. 
Two matrices $A, B \in \mathbb{R}^{n \times n}$ are said to be similar if there exists a nonsingular matrix $S \in \mathbb{R}^{n \times n}$ such that $A = S^{-1} B S$.
In particular, if matrices $A$ and $B$ are similar, then they have the same set of eigenvalues.  
We use the following simple fact involving Loewner ordering: given two invertible matrices $A$ and $B$ satisfying $\frac{1}{\alpha} B \preceq A \preceq \alpha B$ for some $\alpha > 0$, we have $\frac{1}{\alpha} B^{-1} \preceq A^{-1} \preceq \alpha B^{-1}$. We further need the following facts. 

\begin{fact}[Generalized Lieb-Thirring Inequality~\cite{e13, alo16, jllpt20}] \label{fact:gen_Lieb}
 Given a symmetric matrix $B$, a positive semi-definite matrix $A$ and $\alpha \in [0, 1]$, we have
\begin{align*}
\tr[A^\alpha B A^{1-\alpha} B] \leq \tr[A B^2]. 
\end{align*}
\end{fact}

\begin{fact}[Hoffman-Wielandt Theorem, \cite{hoffman1953, hj12}]\label{fac:hw_thm} Let $A, E \in \R^{n\times n}$ such that $A$ and $A+ E$ are both normal matrices. Let $\lambda_1, \lambda_2, \hdots, \lambda_n$ be the eigenvalues of $A$, and let $\widehat{\lambda}_1, \widehat{\lambda}_2, \hdots, \widehat{\lambda}_n$ be the eigenvalues of $A+E$ in any order. There is a permutation $\sigma$ of the integers $1, \hdots, n$ such that $\sum_{i\in [n]} |\widehat{\lambda}_{\sigma(i)} - \lambda_i |^2 \leq \| E \|_F^2 = \tr [E^* E]$. 
\end{fact}
\begin{fact}[Corollary of the Hoffman-Wielandt Theorem, \cite{hj12}]\label{fac:hoffmanwielandt} Let $A, E \in \R^{n \times n}$ such that $A$ is Hermitian and $A+ E $ is normal. Let $\lambda_1, \dotsc, \lambda_n$ be the eigenvalues of $A$ arranged in increasing order $\lambda_1 \leq \dotsc \leq  \lambda_n$. Let $\widehat{\lambda}_1, \dotsc, \widehat{\lambda}_n$ be the eigenvalues of $A+ E$, ordered so that $\mathrm{Re}(\widehat{\lambda}_1) \leq \dotsc \leq \mathrm{Re}(\widehat{\lambda}_n)$. Then, $\sum_{i\in [n]} |\widehat{\lambda}_i - \lambda_i|^2 \leq \| E \|_F^2$. 
\end{fact}
\begin{fact}[Woodbury matrix identity, \cite{w49,w50}]\label{fac:woodbury} 
Given matrices $A\in \R^{n\times n}$, $U\in \R^{n \times k}$, $C\in \R^{k \times k}$, and $V \in \R^{k \times n}$, such that $A$, $C$, and $A+ UCV$ are invertible, we have 
\begin{align*}
(A + UCV)^{-1} = A^{-1} - A^{-1} U (C^{-1} + VA^{-1} U)^{-1} V A^{-1}.
\end{align*}
\end{fact}

 %%% Section 2. Preliminaries

\section{Matrix Multiplication}\label{subsec:matproperty}

The main goal of this section is to derive upper bounds on the time to perform the following two rectangular matrix multiplication tasks (Lemma~\ref{lem:mn2_less_nmn},~\ref{lem:mn2_trueomega}, and~\ref{lem:mn2_cheaper}): 
\begin{itemize}
	\vspace{-0.2cm}
\item Multiplying a matrix of dimensions $m \times n^2$ with one of dimensions $n^2 \times m$,
\vspace{-0.2cm}
\item Multiplying a matrix of dimensions $n \times mn$ with one of dimensions $mn \times n$. 
\end{itemize} 
\vspace{-0.2cm}
Besides being crucial to the runtime analysis of our interior point method in Section \ref{sec-cost}, these results (as well as several intermediate results) might be of independent interest. 

\subsection{Exponent of matrix multiplication}

We need the following definitions to describe the cost of certain fundamental matrix operations we use. 

\begin{definition}\label{def:cmat}
Define $\Tmat(n,r,m)$ to be the number of operations needed to compute the product of matrices of dimensions $n \times r$ and $r \times m$.
\end{definition}

\begin{definition}\label{def:omegak} 
We define the function $\omega(k)$ to be the
minimum value such that $\Tmat( n, n^k, n ) = n^{ \omega( k ) + o( 1 ) } $. We overload notation and use $\omega$ to denote  the cost of multiplying two $n\times n$ matrices. Thus, we have $\omega(1)= \omega$. 
\end{definition}

The following is a basic property of $\Tmat$ that we frequently use. 

\begin{lemma}[\cite{bcs97,b13}]\label{lem:OrderOfTmat}
For any three positive integers $n,m,r$, we have
\begin{align*}
\Tmat(n,r,m) = O(\Tmat(n,m,r)) = O(\Tmat(m,n,r)).
\end{align*}
\end{lemma}

We refer to Table 3 in \cite{gu18} for the latest upper bounds on $\omega(k)$ for different values of $k$.  In particular, we need the following upper bounds in our paper. 

\begin{lemma}[\cite{gu18}] \label{lem:omega2_bound}
We have: 
\begin{itemize}
\vspace{-0.2cm}
\item  $\omega = \omega(1) \leq 2.372927$, 
\vspace{-0.2cm}
\item  $\omega (1.5) \leq 2.79654$, 
\vspace{-0.2cm}
\item  $\omega (1.75) \leq 3.02159$,
\vspace{-0.2cm}
\item  $\omega(2) \leq 3.251640$. 
\end{itemize}
\end{lemma}

\subsection{Technical results for matrix multiplication}
In this section, we derive some technical results on $\Tmat$ and $\omega$ that we extensively use for our runtime analysis. Some of these results can be derived using tensors, and we demonstrate this in Appendix \ref{subsec:tensormatproperty}. We hope that the use of tensors can yield better runtimes for this problem in future. 
%%%%%%%%%%%%%%%%%%%%%%%%%%%%%%%%%%%
%%%%%%%%%%%%%%%%%%%%%%%%%%%%%%%%%%

\begin{lemma}[Sub-linearity] \label{lem:sublinearity}
For any $p \geq q \geq 1$, we have
\begin{align*}
\omega(p) \leq p - q + \omega(q).
\end{align*}
\end{lemma}

\begin{proof}
We assume that $n^p$ and $n^q$ are integers for notational simplicity. 
Consider multiplying an $n \times  n^p$ matrix with an $ n^p  \times n$ matrix. 
One can cut the $n \times n^p$ matrix into $n^{p-q}$ rectangular blocks of size $n \times n^q$ and the $n^p \times n$ matrix into $n^{p-q}$ rectangular blocks of size $n^q \times n$, and compute the multiplication of the corresponding blocks. 
This approach takes time $n^{p-q + \omega(q) + o(1)}$, from which the desired inequality immediately follows. 
\end{proof}

Key to our analysis is the following lemma, which establishes the convexity of $\omega(k)$. 

%%%%%%%%%%%%%%%%%%%%%%%%%%%%%%%%%%
%%%%%%%%%%%%%%%%%%%%%%%%%%%%%%%%%%

\begin{lemma}[Convexity]\label{lem:omegaconvex} 
The fast rectangular matrix multiplication time exponent $\omega(k)$ as defined in Definition~\ref{def:omegak} is convex in $k$.
\end{lemma}

\begin{proof} Let $k = \alpha \cdot p + (1-\alpha) \cdot q$ for $\alpha \in (0,1)$. For notational simplicity, we assume that $n^p$, $n^q$ and $n^k$ are all integers. 
Consider a rectangular matrix of dimensions $n \times n^k$. Since $\alpha  p \leq k$, we can tile this rectangular matrix with matrices of dimensions $n^{\alpha} \times n^{\alpha  p}$. Then, the  product of this tiled matrix with another similarly tiled matrix of dimensions $n^k \times n$ can be obtained by viewing it as a multiplication of a matrix of dimensions $n/n^{\alpha} \times n^{k}/n^{\alpha  p}$ with one of dimensions $n^{k}/n^{\alpha  p} \times n^{1/\alpha}$, where each ``element'' of these two matrices is itself a matrix of dimensions $n^{\alpha} \times n^{\alpha  p}$. With this recursion in tow, we obtain the following upper bound. 
\begin{align*}
\Tmat(n, n^k, n) 
\leq & \Tmat(n^{\alpha}, n^{\alpha p}, n^{\alpha}) \cdot \Tmat(n/n^\alpha, n^k/n^{\alpha  p}, n / n^{\alpha}) \\
= & \Tmat(n^{\alpha}, n^{\alpha  p}, n^{\alpha}) \cdot \Tmat(n^{(1-\alpha)}, n^{(1-\alpha) q}, n^{(1-\alpha)})\\
\leq & n^{\alpha \cdot \omega(p) + o(1)} \cdot n^{(1-\alpha) \cdot \omega(q) + o (1)} .
\end{align*} The final step above follows from denoting $m = n^\alpha$ and observing that multiplying matrices of dimensions $n^\alpha \times n^{\alpha \cdot p}$ costs, by Definition~\ref{def:omegak}, $m^{\omega(p) + o(1)}$, which is exactly $n^{\alpha (\omega (p) + o(1))}$.  Applying Definition~\ref{def:omegak} and comparing exponents, this implies that 
\begin{align*}
\omega(k) &\leq \alpha\cdot\omega(p) + (1-\alpha)\cdot \omega(q) ,
\end{align*} 
which proves the convexity of the function $\omega(k)$. 
\end{proof}
%%%%%%%%%%%%%%%%%%%%%%%%%%%%%%%%%%
%%%%%%%%%%%%%%%%%%%%%%%%%%%%%%%%%%

\begin{claim}\label{claim:omegaval168} 
$\omega(1.68568)\leq 2.96370$.
\end{claim}
\begin{proof}
We can upper bound $\omega(1.68568)$ in the following sense 
\begin{align*}
\omega(1.68568) 
= & ~ \omega(0.25728 \cdot 1.5 + (1 - 0.25728) \cdot 1.75) \\
\leq & ~ 0.25728 \cdot \omega(1.5) + (1 - 0.25728) \cdot \omega(1.75) \\
\leq & ~ 0.25728 \cdot 2.79654 + (1- 0.25728) \cdot 3.02159\\
\leq & ~ 2.96370,
\end{align*}  
where the first step follows from convexity of $\omega$ (Lemma~\ref{lem:omegaconvex}), the third step follows from $\omega(1.5) \leq 2.79654$ and $\omega(1.75) \leq 3.02159$ (Lemma~\ref{lem:omega2_bound}).
\end{proof}
%%%%%%%%%%%%%%%%%%%%%%%%%%%%%%%%%%
%%%%%%%%%%%%%%%%%%%%%%%%%%%%%%%%%%

\begin{lemma} \label{lem:hlk_less_hkl}
Let $\Tmat$ be defined as in Definition~\ref{def:cmat}. Then for any positive integers $h$, $\ell$, and $k$, we have
\begin{align*}
\Tmat(h, \ell k, h) \leq O(\Tmat(h k , \ell , h k)).
\end{align*}
\end{lemma}

\begin{proof}
Given any matrices $A, B^\top \in \mathbb{R}^{h, \ell k}$, by Definition~\ref{def:cmat}, the cost of computing the matrix product $AB$ is $\Tmat(h, \ell k, h)$. We now show how to compute this product in time $O(\Tmat(hk, \ell, hk))$. 
We cut $A$ and $B^\top$ into $k$ sub-matrices each of size $h \times \ell$, i.e. $A = (A_1,\cdots, A_k)$ and $B^\top = (B_1^\top, \cdots, B_k^\top)$, where each $A_i, B_i^\top \in \mathbb{R}^{h \times \ell}$ for all $i \in [k]$. By performing matrix multiplication blockwise, we can write  
\begin{align*}
AB = \sum_{i=1}^k A_i B_i.
\end{align*}
Next, we stack the $k$ matrices $A_1,\cdots,A_k$ vertically to form a matrix $A' \in \mathbb{R}^{hk,\ell}$. Similarly, we stack the $k$ matrices $B_1, \cdots, B_k$ horizontally to form a matrix $B' = (B_1,\cdots,B_k) \in \mathbb{R}^{\ell, hk}$. 
By Definition~\ref{def:cmat}, we can compute $A' B' \in \mathbb{R}^{hk, hk}$ in time $\Tmat(hk, \ell, hk)$. To complete the proof, we note that we can derive $AB$ from $A' B'$ as follows: for each $j \in [k]$, the $j$th diagonal block of $A' B'$ of size $h \times h$ is exactly $A_j B_j$, and summing up the $k$ diagonal $h \times h$ blocks of $A' B'$ gives $AB$. 
\end{proof}

%%%%%%%%%%%%%%%%%%%%%%%%%%%%%%%%%%
%%%%%%%%%%%%%%%%%%%%%%%%%%%%%%%%%%
\subsection{General upper bound on $\Tmat(n, mn, n)$ and $\Tmat(m, n^2, m)$}

\begin{lemma}\label{lem:mn2_less_nmn}
Let $\Tmat$ be defined as in Definition~\ref{def:cmat}. \\
If $m \geq n$, then we have
\begin{align*}
\Tmat(n, m n, n) \leq O(\Tmat(m, n^2, m)) .
\end{align*}
If $m \leq n$, then we have
\begin{align*}
\Tmat(m, n^2, m) \leq O(\Tmat(n, m n, n)) .
\end{align*}
\end{lemma}

\begin{proof}
We only prove the case of $m \geq n$, as the other case where $m < n$ is similar. 
This is an immediate consequence of Lemma~\ref{lem:hlk_less_hkl} by taking $h = n$, $\ell = n^2$, and $k = \lfloor m/n \rfloor$, where $k$ is a positive integer because $m \geq n$.
\end{proof}

In the next lemma, we derive upper bounds on the term $\Tmat(m, n^2, m)$ when $m \geq n$ and $\Tmat(n, mn, n)$ when $m < n$, which is crucial to our runtime analysis. 

%%%%%%%%%%%%%%%%%%%%%%%%%%%%%%%%%%
%%%%%%%%%%%%%%%%%%%%%%%%%%%%%%%%%%

\begin{lemma}\label{lem:mn2_trueomega} 
Let $\Tmat$ be defined as in Definition~\ref{def:cmat} and $\omega$ be defined as in Definition~\ref{def:omegak}. \\
Property I. We have 
\begin{align*}
\Tmat(n, mn, n) \leq O (m n^{\omega + o(1)}) .
\end{align*}
Property II. We have
\begin{align*}
\Tmat(m, n^2, m) \leq O \left(\sqrt{n} \left(mn^2  + m^{\omega} \right) \right) .
\end{align*}
%for the true value of $\omega(1)$. 
\end{lemma}
\begin{proof}

{\bf Property I.}

Recall from Definition~\ref{def:cmat} that $\Tmat(n, mn, n)$ is the cost of multiplying a matrix of size $n \times mn$ with one of size $mn \times n$. We can cut each of the matrices into $m$ sub-matrices of size $n \times n$ each. The product in question then can be obtained by multiplying these sub-matrices. Since there are $m$ of them, and each product of an $n\times n$ submatrix with another $n \times n$ submatrix costs, by definition, $n^{\omega + o(1)}$, we get $\Tmat(n, mn, n) \leq O(mn^{\omega + o(1)})$, as claimed. 

{\bf Property II.}

Let $m = n^a$, where $a \in (0,\infty)$. By definition, $\Tmat(m, n^2, m)$ is the cost of multiplying a matrix of size $m \times n^2$ with one of size $n^2 \times m$. Expressing $n^2$ as $m^{2/a}$ then gives, by Definition~\ref{def:omegak}, that 
\begin{align*}
\Tmat(m, n^2, m) = m^{\omega(2/a) + o(1)} = n^{a \cdot \omega(2/a) + o(1)}
\end{align*}
 Property II is then an immediate consequence of the following inequality, which we prove next: 
\begin{align}
\omega(2/a) < \max(1 + 2.5/a, \omega(1) + 0.5/a) \qquad \forall a \in (0,\infty). \label{eq:omega-2-by-a-bound}
\end{align}
Define $b = 2/a \in (0, \infty)$. Then the desired inequality in \eqref{eq:omega-2-by-a-bound} can be expressed in terms of $b$ as 
\begin{align}\label{tbd-lines}
\omega(b) < \max( 1 + 5b/4, \omega(1) + b/4) \qquad \forall b \in (0, \infty) .
\end{align}
Notice that the RHS of \eqref{tbd-lines} is a maximum of two linear functions of $b$ and these intersect at $b^* = \omega(1) - 1$. 
By the convexity of $\omega({}\cdot{})$ as proved in  Lemma~\ref{lem:omegaconvex}, it suffices to verify \eqref{tbd-lines} at the endpoints $b \rightarrow 0$, $b \rightarrow \infty$ and $b = b^*$.
%check the endpoints of the interval of maximization and the point of intersection of the two lines, that is, the points $b = 1$, $b = 2$, and $b = b^*$. 
In the case where $b = \delta$ for any $\delta < 1$, \eqref{tbd-lines} follows immediately from the observation that $\omega(\delta) < \omega(1)$. 
We next argue about  the case $b \rightarrow \infty$. By Lemma \ref{lem:omega2_bound} we have $\omega(2) \leq 3.252$. Using Lemma \ref{lem:sublinearity}, we have $\omega(b)\leq b - 2 + \omega(2)$. Combining these two facts implies that for any $b > 2$, we have 
\begin{align*}
\omega(b) \leq b - 2 + \omega(2) \leq 1 + 5b / 4, 
\end{align*} which again satisfies \eqref{tbd-lines}. The final case is $b = b^* = \omega(1) - 1$, for which \eqref{tbd-lines} is equivalent to 
\begin{align}\label{tbd-bstar}
\omega(\omega(1) - 1) < 5\omega(1)/4 - 1/4.
\end{align}
By Lemma \ref{lem:omega2_bound}, we have that $\omega(1) - 2 \in [0, 0.372927]$.
Then to prove \eqref{tbd-bstar}, it is sufficient to show that 
\begin{align}\label{tbd-t}
\omega(t + 1) < 5t/4 + 9/4 \qquad \forall t \in [0, 0.372927] .
\end{align}
By the convexity of $\omega({}\cdot{})$ as proved in Lemma \ref{lem:omegaconvex}, the upper bound of $\omega(2) \leq 3.251640$ in Lemma \ref{lem:omega2_bound}, and recalling that $\omega(1) = t + 2$ for $t\in [0, 0.372927]$, we have for $k \in [1, 2]$, 
\begin{align*}
\omega(k) \leq \omega(1) + (k-1) \cdot (3.251640 - (t+2)) = t + 2 + (k-1) \cdot (1.251640 - t).
\end{align*} 
In particular, using this inequality for $k = t+1$, we have 
\begin{align*}
	\omega(t + 1) - 5t/4 - 9/4 &\leq (t + 2) + t\cdot (1.251640 - t) - 5t/4 - 9/4 \\
	&= -t^2 + 1.00164 t - 1/4,
\end{align*} 
which is negative on the entire interval $[0,0.372927]$. 
This establishes \eqref{tbd-t} and finishes the proof. 
\end{proof}

%%%%%%%%%%%%%%%%%%%%%%%%%%%%%%%%%%
%%%%%%%%%%%%%%%%%%%%%%%%%%%%%%%%%%

\subsection{Specific upper bound on $\Tmat(m,n^2,m)$}
% NOTE: We will revisit this conjecture after the FOCS deadline; commenting out just for the submission.
%\haotian{If we can't prove this, this should be removed eventually.}
%
%We hope to show:
%\begin{conjecture} 
%Let $\Tmat$ be defined as in \ref{def:cmat}. 
%For any two positive integers $n,m$ satisfying that $n \leq m \leq n^2$, we have
%\begin{align*}
%\Tmat(m, n^2, m) = o \left(m^3 + m n^{\omega} \right) .
%\end{align*}
%\end{conjecture}

\begin{comment}
\Zhao{This is possible}
\begin{lemma}
For any two positive integers $n,m$ satisfying that $n \leq m \leq n^2$, we have
\begin{align*}
\Tmat(m, n^2, m) = m^{ 2 \omega( 1 ) - \omega( 2 ) } n^{ 2 ( \omega( 2 ) - \omega( 1 ) ) } .
\end{align*}
Further, it can be simplified to
\begin{align*}
\Tmat(m, n^2, m) = O(m^{ 2 \omega (1) - 3 } n^2)
\end{align*}
\end{lemma}

\begin{proof}

Part 1. TODO

Part 2. Since we have
\begin{align*}
3 \leq \omega(2) \leq \omega(1) + 1 
\end{align*}
Then
\begin{align*}
2\omega (1) - \omega(2) \leq 2 \omega(1) -3
\end{align*}
and
\begin{align*}
2 (\omega(2) - \omega(1)) \leq 2 \cdot 1 = 2
\end{align*}
\end{proof}
Currently, we can only prove:

\end{comment}

%%%%%%%%%%%%%%%%%%%%%%%%%%%%%%%%%%
%%%%%%%%%%%%%%%%%%%%%%%%%%%%%%%%%%

\begin{lemma}\label{lem:mn2_cheaper} 
For any two positive integers $n$ and $m$, we have
\begin{align*}
\Tmat(m, n^2, m) = o \left(m^3 + m n^{2.37} \right) .
\end{align*}
\end{lemma}
\begin{proof}

 Let $m = n^a$ where $a \in (0,\infty)$. Recall  that $\Tmat(m,n^2, m) = m^{\omega(2/a)+o(1)} =  n^{a \omega(2/a) + o(1)}$. We consider the following two cases according to the range of $a$. 

 \textbf{Case 1: $a \in [ 1.18647, \infty)$}. In this case, we have $\omega(2/a) \leq \omega(2/1.18647) \leq \omega(1.68568) < 3$, where the last inequality follows from Claim~\ref{claim:omegaval168}. This implies that 
\begin{align}
\Tmat(m, n^2, m) = o(n^{3a}) = o(m^3). \label{eq:case1-tmat-mn2}
\end{align} 

 \textbf{Case 2: $a \in (0, 1.18647]$}. In this case, we have $ 2/a \in [1.68567, \infty)$. 
Consider the linear function
\begin{align}\label{eq-lineupperbound}
y(t) = 1 + 2.37 \cdot \frac{t}{2}. 
\end{align} 
By Claim~\ref{claim:omegaval168}, we have 
\begin{align}
\omega(1.68567) < 2.997 \leq y(1.68567). \label{eq:omega-y-left}
\end{align} By Lemma \ref{lem:omega2_bound}, we have 
\begin{align}
\omega(2) < 3.37 = y(2). \label{eq:omega-y-right} 
\end{align} An application of Lemma \ref{lem:sublinearity} then gives, for any $t\geq 2$, the inequality 
\begin{align} 
\omega(t) \leq t - 2 + \omega(2) < t - 2 + y(2) \leq  y(t),\label{eq:omega-y-beyondright}
\end{align} where the last inequality is by definition of $y(t)$ from \eqref{eq-lineupperbound}.  
Therefore, combining the convexity of $\omega({}\cdot{})$, as proved in Lemma~\ref{lem:omegaconvex}, with \eqref{eq:omega-y-left}, \eqref{eq:omega-y-right}, and \eqref{eq:omega-y-beyondright}, we conclude that for any $t \in [1.68567, \infty)$, the function $\omega$ is bounded from above by the affine function $y$, expressed as follows. $$\omega(t) < y(t) =  1 + 2.37 \cdot \frac{t}{2}.$$ 
This implies that 
\begin{align}
\Tmat(m, n^2, m) = n^{a \cdot \omega(2/a) + o(1)} = o(n^{a +  2.37}) = o(mn^{3.27}). \label{eq:case2-tmat-mn2}
\end{align}
Combining the results from \eqref{eq:case1-tmat-mn2} and \eqref{eq:case2-tmat-mn2} finishes the proof of the lemma. 
\end{proof}

 %%% Section 3. Matrix multiplication

%\newpage
\section{Main Theorem}

In this section, we give the formal statement of our main result. 

\begin{theorem}[Main result, formal]\label{thm:main_formal}
Consider a semidefinite program with variable size $n \times n$ and $m$ constraints (assume there are no redundant constraints): 
\begin{align} \label{eq:SDP_main_formal}
\max \langle C, X \rangle  \textup{ subject to } X \succeq 0, \langle A_i, X \rangle = b_i \textup{ for all } i \in [m]. 
\end{align} 
Assume that any feasible solution $X \in \mathbb{S}^{n \times n}_{\geq 0}$ satisfies $\norm{X}_\op \leq R$. 
Then for any error parameter $0 < \delta \leq 0.01$, 
there is an interior point method that outputs in time $O^*(\sqrt{n}(mn^2 + m^\omega + n^\omega) \log(n / \delta))$
a positive semidefinite matrix $X \in \mathbb{R}^{n \times n}_{\geq 0}$ such that 
\begin{align*}
\langle C, X \rangle \geq \langle C, X^* \rangle - \delta \cdot \norm{C}_\op \cdot R \quad \text{and} \quad \sum_{i \in [m]}\left| \langle A_i, \widehat{X} \rangle - b_i \right| \leq 4n\delta \cdot (R  \sum_{i \in [m]} \norm{A_i}_1 + \norm{b}_1) ,
\end{align*}
where $\omega$ is the exponent of matrix multiplication, $X^*$ is any optimal solution to the semidefinite program in \eqref{eq:SDP_main_formal}, and $\norm{A_i}_1$ is the Schatten $1$-norm of matrix $A_i$. 
\end{theorem}

The proof of Theorem~\ref{thm:main_formal} is given in the subsequent sections.

%%%%%%%%%%%%%%%%%%%%%%%%%%%%%%%%
%%%%%%%%%%%%%%%%%%%%%%%%%%%%%%%%

\section{Approximate Central Path via Approximate Hessian}\label{sec-centralpath}

\subsection{Main result for approximate central path} 
Our main result of this section is the following.

%%%%%%%%%%%%%%%%%%%%%%%%%%%%%%%%
%%%%%%%%%%%%%%%%%%%%%%%%%%%%%%%%

\begin{theorem}[Approximate central path]\label{thm:approx-central-path}
Consider a semidefinite program as in Definition~\ref{defn:sdpprimal} with no redundant constraints. Assume that any feasible solution $X \in \mathbb{S}^{n \times n}_{\geq 0}$ satisfies $\norm{X}_\op \leq R$. Then for any error parameter $0 < \delta \leq 0.01$ and Newton step size $\epsilon_N$ satisfying  $\sqrt{\delta} < \epsilon_N \leq 0.1$, Algorithm~\ref{alg:ss_ipm} outputs,  in $T = \frac{40}{\epsilon_N} \sqrt{n} \log(n/\delta)$ iterations,  a positive semidefinite matrix $X \in \mathbb{R}^{n \times n}_{\geq 0}$ that satisfies 
\begin{align}\label{eq:approx_optimality}
\langle C, X \rangle \geq \langle C, X^* \rangle - \delta \cdot \norm{C}_\op \cdot R \quad \text{and} \quad \sum_{i \in [m]}\left| \langle A_i, \widehat{X} \rangle - b_i \right| \leq 4n\delta \cdot (R  \sum_{i \in [m]} \norm{A_i}_1 + \norm{b}_1) ,
\end{align}
where $X^*$ is any optimal solution to the semidefinite program in Definition~\ref{defn:sdpprimal}, and $\norm{A_i}_1$ is the Schatten $1$-norm of matrix $A_i$. Further, in each iteration of Algorithm~\ref{alg:ss_ipm}, the following invariant holds for $\alpha_H = 1.03$: 
\begin{align}\label{eq:promise_sdp}
\| S^{-1/2} S_{\new} S^{-1/2} - I \|_F \leq \alpha_H \cdot \epsilon_N. 
\end{align}
\end{theorem}

\begin{proof}
At the start of Algorithm~\ref{alg:ss_ipm}, Lemma~\ref{lem:init} is called to modify the semidefinite program to obtain an initial dual solution $y$ for the modified SDP that is close to the dual central path at $\eta = 1/(n+2)$. This ensures that the invariant $g_\eta(y)^\top H(y)^{-1} g_\eta(y) \leq \epsilon_N^2$ holds at the start of the algorithm. Therefore, by Lemma~\ref{lem:hessapprox} and Lemma~\ref{lem:invariant_newton}, this invariant continues to hold throughout the run of the algorithm. Therefore, after $T = \frac{40}{\epsilon_N} \sqrt{n} \log\left( \frac{n}{\delta} \right)$ iterations, the step size $\eta$ in Algorithm~\ref{alg:ss_ipm} grows to $\eta = (1+\frac{\epsilon_N}{20\sqrt{n}})^T / (n+2) \geq 2 n / \delta^2$. It then follows from Lemma~\ref{lem:approximate_optimality} that 
\begin{align*}
b^\top y \leq b^\top y^* + \frac{n}{\eta} \cdot (1 + 2 \epsilon_N) \leq b^\top y^* + \delta^2 .
\end{align*}
Thus when the algorithm stops, the dual solution $y$ has duality gap at most $\delta^2$ for the modified SDP. 
Lemma~\ref{lem:init} then shows how to obtain an approximate solution to the original SDP that satisfies the guarantees in \eqref{eq:approx_optimality}.

To prove~\eqref{eq:promise_sdp}, define $\Delta_S =  S_{\new} - S \in \R^{n \times n}$ and $\delta_y = y_{\new} - y \in \R^m$. For each $i\in [n]$, we use $\delta_{y,i}$ to denote the $i$-th coordinate of vector $\delta_y$.  
We rewrite $\| S^{-1/2} S_{\new} S^{-1/2} -I \|_F^2 $ as
\begin{align*}
\| S^{-1/2} S_{\new} S^{-1/2} -I \|_F^2 
= & ~ \tr \Big[ (S^{-1/2} (\Delta_S) S^{-1/2})^2 \Big]\\
= & ~ \tr \left[ S^{-1} \left(\sum_{i=1}^m \delta_{y,i} A_i \right) S^{-1} \left(\sum_{j=1}^m \delta_{y,j} A_j \right) \right]\\
= & ~ \sum_{i, j=1}^m \delta_{y,i} \delta_{y,j} \tr[ S^{-1} A_i S^{-1} A_j ] \\
= & ~ (\delta_y)^\top H(y) \delta_y \\
= & ~ g_{\eta}(y)^\top \wt{H}(y)^{-1} H(y) \wt{H}(y)^{-1} g_{\eta}(y), \numberthis\label{s-change-1}
\end{align*} 
where we used the fact that $\Delta_S = \sum_{i=1}^m (\delta_y)_i A_i$. 
It then follows from Lemma~\ref{lem:hessapprox} and the invariant $g_\eta(y)^\top H(y)^{-1} g_\eta(y) \leq \epsilon_N^2$ that
\[
g_{\eta}(y)^\top \wt{H}(y)^{-1} H(y) \wt{H}(y)^{-1} g_{\eta}(y) \leq \alpha_H^2 \cdot \epsilon_N^2 , \numberthis\label{h-change-1}
\]
where $\alpha_H = 1.03$. Combining Equation~\eqref{s-change-1} with Inequality~\eqref{h-change-1} completes the proof of the theorem.
\end{proof}

\begin{table}[htp!]\caption{Summary of parameters in approxiate central path.}
\centering
  \begin{tabular}{ | l | l | l | l | }
    \hline
    {\bf Notation} & {\bf Choice} & {\bf Appearance} & {\bf Meaning} \\ \hline
    $\alpha_H$ & 1.03 & Lemma~\ref{lem:hessapprox} & Spectral approximation factor $\alpha_H^{-1} \cdot H \preceq \wt{H} \preceq \alpha_H \cdot H$ \\ \hline
    $\epsilon_N$ & 0.1 & Lemma~\ref{lem:invariant_newton} &  Upper bound on the Newton step size $(g_\eta^\top H^{-1} g_\eta)^{1/2}$ \\ \hline
    $\epsilon_S$ & 0.01 & Algorithm~\ref{alg:approx_slack_update} & Spectral approximation error $(1-\epsilon_S) \cdot S \preceq \wt{S} \preceq (1 + \epsilon_S) \cdot S$ \\ \hline
  \end{tabular} \label{Tab:params}
\end{table}

\begin{algorithm}[!t]
\caption{ }\label{alg:ss_ipm}
\begin{algorithmic}[1]
\Procedure{\textsc{Main}}{$n,m,\delta, \epsilon_N, C,A,b$}
\Comment{$C \in \mathbb{S}^{n \times n}$, $\{A_i\}_{i=1}^m \in \mathbb{S}^{n \times n}$, vector $b \in \R^m$, error parameter $0 < \delta < 0.1$, Newton step size parameter $0< \epsilon_N < 0.1$}
\State Modify the SDP and obtain an initial dual solution $y$ according to Lemma~\ref{lem:init}
\State $\eta \leftarrow 1/(n+2)$
\State $T \leftarrow \frac{40}{\epsilon_N} \sqrt{n}\log \left( \frac{n}{\delta}\right)$
\State  $\wt{S} \leftarrow S \leftarrow \sum_{i \in [m]} y_i A_i - C$.
\For {$\text{iter} = 1 \to T$}
	\State $\eta_{\new} \leftarrow \eta \left( 1 + \frac{\epsilon_N}{20 \sqrt{n}}\right)$
	\For {$j = 1, \cdots, m$} \Comment{Gradient computation} 
		\State $g_{\eta_{\new}}(y)_j \leftarrow  \eta_{\new} \cdot b_j -   \tr [ S^{-1} \cdot A_j ] $ 
	\EndFor
	\For {$j = 1, \cdots, m$} \Comment{Hessian computation} \label{lin:start_compute_Hessian}
		\For {$k = 1, \cdots, m$} 
			\State $\wt{H}_{j,k}(y) \leftarrow \tr [ \wt{S}^{-1} \cdot A_j \cdot \wt{S}^{-1}\cdot A_k ]$ 
		\EndFor 
	\EndFor \label{lin:end_compute_Hessian}
	\State $\delta_y \leftarrow - \wt{H}(y)^{-1} g_{\eta_{\new}}(y)$ \Comment{Update on $y$}
	\State $y_{\new} \leftarrow y + \delta_y$ \Comment{Approximate Newton step}
	\State $S_{\new} \leftarrow \sum_{i \in [m]} (y_{\new})_i A_i - C$ 
	\State $\wt{S}_{\new} \leftarrow \textsc{ApproxSlackUpdate}(S_{\new}, \wt{S})$ \label{step:slack-update} \Comment{Approximate slack computation}
	\State $y \leftarrow y_{\new}$,  $S \leftarrow S_{\new}$, $\wt{S} \leftarrow \wt{S}_{\new}$ \Comment{Update variables}
\EndFor
\State Return an approximate solution to the original SDP according to Lemma~\ref{lem:init}
\EndProcedure
\end{algorithmic}
\end{algorithm}

\begin{algorithm}[!t]
\caption{Approximate Slack Update}\label{alg:approx_slack_update}
\begin{algorithmic}[1]
\Procedure{\textsc{ApproxSlackUpdate}}{$S_{\new}, \wt{S}$} 
\Comment{$S_{\new}, \wt{S} \in \mathbb{S}^{n \times n}_{\geq 0}$ are positive definite matrices}
\State $\epsilon_S \leftarrow 0.01$ \Comment{Spectral approximation constant}
\State $Z_{\mathrm{mid}} \leftarrow S_{\new}^{-1/2} \cdot \wt{S} \cdot S_{\new}^{-1/2} - I$ 
\State Compute spectral decomposition $Z_{\mathrm{mid}} = U \cdot \Lambda \cdot U^\top$ 
\State \Comment{$\Lambda = \diag(\lambda_1,\cdots, \lambda_n)$ are the eigenvalues of $Z_{\mathrm{mid}}$, and $U \in \R^{n \times n}$ is orthogonal}
\State Let $\pi: [n] \rightarrow [n]$ be a sorting permutation such that $|\lambda_{\pi(i)}| \geq |\lambda_{\pi(i+1)}|$
\If {$|\lambda_{\pi(1)}| \leq \epsilon_S$}
	\State $\wt{S}_{\new} \leftarrow \wt{S}$ 
\Else
	\State $r \leftarrow 1$
	\While{$|\lambda_{\pi(2r)}| > \epsilon_S$ or $|\lambda_{\pi(2r)}| > (1 - 1/\log n) |\lambda_{\pi(r)}|$} 
		\State $r \leftarrow r + 1$
	\EndWhile

	\State $(\lambda_{\new})_{\pi(i)} \leftarrow 
		\begin{cases}
		0 & \text{~if~} i = 1, 2, \cdots, 2r; \\
		\lambda_{\pi(i)} & \text{~otherwise.}
		\end{cases}$

	\State $\wt{S}_{\new} \leftarrow \wt{S} + S_{\new}^{1/2} \cdot U \cdot \diag (\lambda_{\new} - \lambda) \cdot U^\top \cdot S_{\new}^{1/2}$ \label{lin:low_rank_update}
\EndIf
\State \Return $\wt{S}_{\new}$
\EndProcedure
\end{algorithmic}
\end{algorithm}

\subsection{Approximate slack update}

%%%%%%%%%%%%%%%%%%%%%%%%%%%%%%%%
%%%%%%%%%%%%%%%%%%%%%%%%%%%%%%%%

\begin{lemma}\label{lem:approx_slack_update}
Given positive definite matrices $S_{\new}, \wt{S} \in \mathbb{S}^{n \times n}_{>0}$ and any parameter $0 < \epsilon_S < 0.01$, there is an algorithm (procedure \textsc{ApproxSlackUpdate} in Algorithm~\ref{alg:approx_slack_update}) that takes $O(n^{\omega+o(1)})$ time to output a positive definite matrix $\wt{S}_{\new} \in \mathbb{S}^{n \times n}_{>0}$ such that
\begin{align}
\| S_{\mathrm{new}}^{-1/2} \widetilde{S}_{\mathrm{new}} S_{\mathrm{new}}^{-1/2} - I \|_{\op} \leq \epsilon_S . \label{approx-slack}
\end{align}
\end{lemma}

\begin{proof}
The runtime of $O(n^{\omega+o(1)})$ is by the spectral decomposition $Z = U \cdot \Lambda \cdot U^\top$, the costliest step in the algorithm. 
To prove~\eqref{approx-slack}, we notice that $\lambda_{\mathrm{new}}$ are the eigenvalues of $S_{\mathrm{new}}^{-1/2} \widetilde{S}_{\mathrm{new}} S_{\mathrm{new}}^{-1/2} - I$ and by the algorithm description (lines 6 - 13), the upper bound $(\lambda_{\mathrm{new}})_i \leq \epsilon_S$ holds for each $i \in [n]$. 
\end{proof}

\subsection{Closeness of slack implies closeness of Hessian}

%%%%%%%%%%%%%%%%%%%%%%%%%%%%%%%%%%%%%%%
%%%%%%%%%%%%%%%%%%%%%%%%%%%%%%%%%%%%%%%

\begin{lemma}\label{lem:appslackappH} 
Given symmetric matrices $A_1,\cdots, A_m \in \mathbb{S}^{n \times n}$, and positive definite matrices $\wt{S}, S \in \mathbb{S}^{n \times n}_{>0}$, define matrices $\wt{H} \in \R^{m \times m}$ and $H \in \R^{m \times m}$ as
\begin{align*}
\wt{H}_{j,k} = \tr [ \wt{S}^{-1} A_j \wt{S}^{-1} A_k ] \qquad \text{and} \qquad H_{j,k} = \tr [ S^{-1} A_j S^{-1} A_k ] .
\end{align*}
Then both $\widetilde{H}$ and $H$ are positive semidefinite. For any accuracy parameter $\alpha_S \geq 1$, if 
\begin{align*}
\alpha_S^{-1} \cdot S \preceq \wt{S} \preceq \alpha_S \cdot S,
\end{align*}
then we have that 
\begin{align*}
\alpha_S^{-2} \cdot H \preceq \wt{H} \preceq \alpha_S^2 \cdot H.
\end{align*} 
\end{lemma}
\begin{proof}
For any vector $v \in \R^n$, we define $A(v) = \sum_{i = 1}^m v_i A_i$. We can rewrite $v^\top H v$ as follows.  
\begin{align}\label{specapprox-1}
v^\top H v = \sum_{i=1}^m \sum_{j=1}^m v_i v_j H_{i,j} = \sum_{i=1}^m \sum_{j=1}^m v_i v_j \tr [ S^{-1} A_i S^{-1} A_j ] = \tr [ S^{-1/2} A(v) S^{-1} A(v) S^{-1/2} ] . 
\end{align}
Similarly, we have
\begin{align}\label{specapprox-2} 
v^\top  \wt{H} v = \tr [ \wt{S}^{-1/2} A(v) \wt{S}^{-1} A(v) \wt{S}^{-1/2} ]. 
\end{align}
As the RHS of \eqref{specapprox-1} and \eqref{specapprox-2} are non-negative, both $\widetilde{H}$ and $H$ are positive semidefinite. 
Since $\wt{S} \preceq \alpha_S \cdot S$, we have $S^{-1} \preceq \alpha_S \cdot \wt{S}^{-1}$ (see Section~\ref{ssec:facts}), which 
gives the following inequalities
\begin{align*}
\tr [ S^{-1/2} A(v) S^{-1} A(v) S^{-1/2} ]
& \leq \alpha_S \cdot \tr [ S^{-1/2} A(v) \wt{S}^{-1} A(v) S^{-1/2} ] \\
& \leq \alpha_S^2 \cdot \tr [ \wt{S}^{-1/2} A(v) \wt{S}^{-1} A(v) \wt{S}^{-1/2} ] , \numberthis\label{Happrox-1}
\end{align*}
where the first inequality follows from viewing $\tr [ S^{-1/2} A(v) S^{-1} A(v) S^{-1/2} ]$ as $\sum_{i=1}^n u_i^\top S^{-1} u_i$ for $u_i =  A(v) S^{-1/2} e_i$ and the second inequality follows similarly, after using the cyclic permutation property of trace. Similarly, using $\alpha_S^{-1} \cdot S \preceq \wt{S}$, we have
\begin{align*}
\tr [ S^{-1/2} A(v) S^{-1} A(v) S^{-1/2} ] \geq \alpha_S^{-2} \cdot \tr [ \wt{S}^{-1/2} A(v) \wt{S}^{-1} A(v) \wt{S}^{-1/2} ] .\numberthis\label{Happrox-2}
\end{align*}
Combining \eqref{Happrox-1} and \eqref{Happrox-2} with \eqref{specapprox-1} and \eqref{specapprox-2} along with the fact that $v$ can  be any arbitrary $n$-dimensional vector finishes the proof of the lemma. 
\end{proof}

%%%%%%%%%%%%%%%%%%%%%%%%%%%%%%%%%%%%
%%%%%%%%%%%%%%%%%%%%%%%%%%%%%%%%%%%

\subsection{Approximate Hessian maintenance}

\begin{lemma}\label{lem:hessapprox} 
In each iteration of Algorithm~\ref{alg:ss_ipm}, for $\alpha_H = 1.03$, the approximate Hessian $\wt{H}(y)$ satisfies that 
\begin{align*}
 \alpha_H^{-1} H(y) \preceq \wt{H}(y) \preceq \alpha_H \cdot H(y). 
\end{align*}
\end{lemma}  

\begin{proof}
By Lemma~\ref{lem:approx_slack_update}, given as input two positive definite matrices $S_{\new}$ and $\wt{S}$, Algorithm~\ref{alg:approx_slack_update} outputs a matrix $\wt{S}_{\new}$ such that 
\begin{align*}
\| S_{\new}^{-1/2} \wt{S}_{\new} S_{\new}^{-1/2} - I \|_{\op} \leq \epsilon_S ,
\end{align*}
where $\epsilon_S = 0.01$ as in Algorithm~\ref{alg:approx_slack_update}. 
By definition of operator norm, this implies that in each iteration of Algorithm~\ref{alg:ss_ipm}, we have, for  $\alpha_S = 1.011$, 
\begin{align*}
\alpha_S^{-1} \cdot S \preceq \wt{S} \preceq \alpha_S \cdot S. 
\end{align*}
The statement of this lemma then follows from Lemma~\ref{lem:appslackappH}.
\end{proof}

\subsection{Invariance of Newton step size}
The following lemma is standard in the theory of interior point methods (e.g. see~\cite{r01}).

%%%%%%%%%%%%%%%%%%%%%%%%%%%%%%%%
%%%%%%%%%%%%%%%%%%%%%%%%%%%%%%%%

\begin{lemma}[Invariance of Newton step~\cite{r01}]\label{lem:invariant_newton}
Given any parameters $1 \leq \alpha_H \leq 1.03$ and $0 < \epsilon_N  \leq 1/10$, suppose that $g_\eta(y)^\top H(y)^{-1} g_\eta(y) \leq \epsilon_N^2$ holds for some feasible dual solution $y \in \mathbb{R}^m$ and parameter $\eta > 0$, and positive definite matrix $\wt{H} \in \mathbb{S}^{n \times n}_{>0}$ satisfies 
\begin{align*}
\alpha_H^{-1} H(y) \preceq \wt{H} \preceq \alpha_H  H(y) % \approx_{\alpha^2}. 
\end{align*} 
Then $\eta_{\new} = \eta (1 + \frac{\epsilon_N}{20 \sqrt{n}} ) $ and $y_{\new} = y - \wt{H}^{-1} g_{\eta_{\new} }(y)$ satisfy 
\begin{align*}
g_{\eta_{\new}} (y_{\new})^\top H(y_{\new})^{-1} g_{\eta_{\new}}(y_{\new}) \leq \epsilon_N^2 .
\end{align*}
\end{lemma}

\subsection{Approximate optimality}

%%%%%%%%%%%%%%%%%%%%%%%%%%%%%%%%
%%%%%%%%%%%%%%%%%%%%%%%%%%%%%%%%

The following lemma is also standard in interior point method.
\begin{lemma}[Approximate optimality~\cite{r01}]\label{lem:approximate_optimality}
Suppose $0 < \epsilon_N \leq 1/10$, dual feasible solution $y \in \mathbb{R}^m$, and parameter $\eta \geq 1$ satisfy the following bound on Newton step size: 
\begin{align*}
g_\eta(y)^\top H(y)^{-1} g_\eta(y) \leq \epsilon_N^2.
\end{align*}
Let $y^*$ be an optimal solution to the dual formulation \eqref{eq:sdpdual}. Then we have
\begin{align*}
b^\top y \leq b^\top y^* + \frac{n}{\eta} \cdot (1 + 2 \epsilon_N) .  
\end{align*}
\end{lemma}

 %%% Section 4. Approximate Central Path via Approximate Hessian

%\newpage

\section{Low-rank Update}\label{sec:lrupdate}
Crucial to being able to efficiently approximate the Hessian in each iteration is the condition that the rank of the update be not too large. We formalize this idea in the following theorem, essential to the runtime analysis in Section~\ref{sec-cost}. 

%%%%%%%%%%%%%%%%%%%%%%%%%%%%%%%%
%%%%%%%%%%%%%%%%%%%%%%%%%%%%%%%%

\begin{theorem}[Rank inequality]\label{thm:rankineq} Let $r_0 = n$ and $r_i$ be the rank of the update to the approximate slack matrix $\wt{S}$ when calling Algorithm~\ref{alg:approx_slack_update} in iteration $i$ of Algorithm~\ref{alg:ss_ipm}. Then, over $T$ iterations of Algorithm~\ref{alg:ss_ipm}, the ranks $r_i$ satisfy the inequality
\begin{align*}
\sum_{i=0}^T \sqrt{r_i} \leq O( T \log^{1.5} n) . 
\end{align*} 
\end{theorem}

The rest of this section is devoted to proving Theorem~\ref{thm:rankineq}. To this end, we define the ``error'' matrix $Z \in \R^{n \times n}$ as follows
\begin{align}
Z = {S}^{-1/2} \wt{S} {S}^{-1/2} - I \label{eq:def-Z-error-matrix}
\end{align}
and the potential function  $\Phi : \R^{n \times n} \rightarrow \R$
\begin{align}\label{eq:def_pot_sdp}
\Phi(Z) = \sum_{i=1}^n \frac{\abs{\lambda(Z)}_{[i]}}{\sqrt{i}}, 
\end{align}
where $\abs{\lambda(Z)}_{[i]}$ denotes the $i$'th entry in the list of absolute eigenvalues of $Z$ sorted in descending order.
The following lemma bounds, from above, the change in the potential described by Equation~\eqref{eq:def_pot_sdp},  when $S$ is updated to $S_{\new}$. 

%%%%%%%%%%%%%%%%%%%%%%%%%%%%%%%%
%%%%%%%%%%%%%%%%%%%%%%%%%%%%%%%%

\begin{lemma}[Potential change when $S$ changes]\label{lem:S_change_pot} 
Suppose matrices $S$, $S_{\new}$ and $\wt{S}$ satisfy the inequalities  
\begin{align}\label{eq:assumpt1}
\| S^{-1/2} S_{\new} S^{-1/2} - I \|_F \leq 0.02 \qquad \text{and} \qquad \| S^{-1/2} \wt{S} S^{-1/2} - I \|_{\mathrm{op}} \leq 0.01 .
\end{align} 
Define matrices $Z = S^{-1/2} \wt{S} S^{-1/2} - I$ and 
$Z_{\mathrm{mid}} = (S_{\new})^{-1/2} \wt{S} (S_{\new})^{-1/2} - I$.
Then we have 
\begin{align*}
\Phi(Z_{\mathrm{mid}}) - \Phi(Z) \leq \sqrt{\log n} .
\end{align*}
\end{lemma}
\begin{proof} 
Our goal is to prove 
\begin{align} \label{eq:Z_close_to_Zmid}
\sum_{i=1}^n (\lambda(Z)_{[i]} -\lambda(Z_{\mathrm{mid}})_{[i]})^2 \leq 10^{-3}.
\end{align}
We first show that the lemma statement is implied by \eqref{eq:Z_close_to_Zmid}. 
We rearrange the order of the eigenvalues of $Z_{\mathrm{mid}}$ and $Z$ so that $\lambda(Z_{\mathrm{mid}})_i$ and $\lambda(Z)_i$ are the $i$th largest eigenvalues of $Z_{\mathrm{mid}}$ and $Z$, respectively. For each $i \in [n]$, denote $\Delta_i = \lambda(Z_{\mathrm{mid}})_{i} - \lambda(Z)_{i}$.  Then \eqref{eq:Z_close_to_Zmid} is equivalent to $\| \Delta \|^2_2 \leq 10^{-3}$. Let $\tau$ be the descending order of the magnitudes of eigenvalues of $Z_{\mathrm{mid}}$, i.e. $\abs{\lambda(Z_{\mathrm{mid}})_{\tau(1)}} \geq \cdots \geq \abs{\lambda(Z_{\mathrm{mid}})_{\tau(n)}}$. The potential change $\Phi(Z_{\mathrm{mid}}) - \Phi(Z)$ can  be upper bounded as
\begin{align*}
\Phi(Z_{\mathrm{mid}}) &= \sum_{i=1}^{n} \frac{1}{\sqrt{i}} \abs{\lambda(Z_{\mathrm{mid}})_{\tau(i)}}\\
&\leq \sum_{i=1}^{n} \left( \frac{1}{\sqrt{i}} \abs{\lambda(Z)_{\tau(i)}} + \frac{1}{\sqrt{i}} |\Delta_{\tau(i)}|\right) \\
&\leq \Phi(Z) + \left( \sum_{i=1}^{n} \frac{1}{i} \right)^{1/2} \left( \sum_{i=1}^{n} |\Delta_i|^2 \right)^{1/2} \\
&\leq \Phi(Z) + \sqrt{\log n}, 
\end{align*} 
where the third line follows from 
\begin{align*}
\sum_i \frac{1}{\sqrt{i}} \abs{\lambda(Z)_{\tau(i)}} \leq \sum_i \frac{1}{\sqrt{i}} \abs{\lambda(Z)}_{[i]} 
\end{align*} 
and Cauchy-Schwarz inequality. This proves the lemma.

The remaining part of this proof is therefore devoted to proving \eqref{eq:Z_close_to_Zmid}. Define $W = S_{\new}^{-1/2} S^{1/2}$. Then, we can express $Z_{\mathrm{mid}}$ in terms of $Z$ and $W$ in the following way. 
\begin{align}\label{eq:ynew}
	Z_{\mathrm{mid}} &= (S_{\new})^{-1/2} \wt{S} (S_{\new})^{-1/2} - I \notag \\
	&= (S_{\new})^{-1/2} S^{1/2} S^{-1/2} \wt{S} S^{-1/2} S^{1/2} (S_{\new})^{-1/2} - I \notag \\
	&= W Z W^\top + WW^\top - I. 
\end{align} 
Let $\lambda(M)_{[i]}$ denote the $i$'th (ordered) eigenvalue of a matrix $M$. We then have 
\begin{align}\label{eq:polyloglemma_int1}
%\sum_{i=1}^n (\lambda(Z_{\mathrm{mid}})_{[i]} - \lambda( Z W^\top W)_{[i]})^2 &=   
 \sum_{i=1}^n (\lambda(Z_{\mathrm{mid}})_{[i]} - \lambda(  W Z W^\top)_{[i]})^2
 &\leq \| Z_{\mathrm{mid}} - W Z W^\top \|_F^2 \notag \\
 %&= \| WW^\top - I \|_F^2 \notag \\
 &= \| W^\top W - I \|_F^2, 
\end{align} 
where the first inequality is by Fact~\ref{fac:hoffmanwielandt} (which is applicable here because $Z_{\mathrm{mid}}$ and $W Z W^\top$ are both normal matrices) and the second step is by \eqref{eq:ynew}. 
Denote the eigenvalues of $S^{-1/2} S_{\new} S^{-1/2}$ by $\{\nu_i\}_{i=1}^n$. Then the first assumption in \eqref{eq:assumpt1} implies that $\sum_{i\in [n]} (\nu_i - 1)^2 \leq 4 \times 10^{-4}$. 
It follows that 
\begin{align}\label{eq:wtw_fnorm}
\| W^\top W - I \|_F^2 = \| S^{1/2} S_{\new}^{-1} S^{1/2} - I \|_F^2 = \sum_{i\in [n]}  (1/\nu_i - 1)^2 \leq 5\times 10^{-4}, 
\end{align} where the last inequality is because the first assumption from \eqref{eq:assumpt1} implies $\nu_i \geq 0.98$ for all $i\in [n]$. Plugging \eqref{eq:wtw_fnorm} into the right hand side of \eqref{eq:polyloglemma_int1}, we have
\begin{align}\label{eq:ynew_yata}
\sum_{i=1}^n (\lambda(Z_{\mathrm{mid}})_{[i]} - \lambda(W Z W^\top)_{[i]})^2 \leq 5 \times 10^{-4}.
\end{align}
Let $W = U \Sigma V^\top$ be the singular value decomposition of $W$, with $U$ and $V$ being $n \times n$ unitary matrices. Because of the invariance of the Frobenius norm under unitary transformation,  \eqref{eq:wtw_fnorm} is then equivalent to
\begin{align}\label{eq:wtw_sos}
\| \Sigma^2 - I \|_F = \sum_{i=1}^n (\sigma_i^2 - 1)^2 \leq 5 \times 10^{-4}.
\end{align}  
Since $U$ and $V$ are unitary, the matrix $W Z W^\top = U \Sigma V^\top Z V \Sigma U^\top$ is similar to $\Sigma V^\top Z V \Sigma$, and the matrix  $Z^\prime = V^\top Z V$ is similar to $Z$. Therefore, 
\begin{align*}
\sum_{i=1}^n (\lambda(W Z W^\top)_{[i]} - \lambda(Z)_{[i]})^2 &= \sum_{i=1}^n (\lambda(\Sigma Z^\prime \Sigma)_{[i]} - \lambda(Z^\prime)_{[i]})^2\\
&\leq \| \Sigma Z^\prime \Sigma - Z^\prime \|_F^2, \numberthis\label{wzwt-z-1}
\end{align*} 
where the last inequality is by Fact~\ref{fac:hoffmanwielandt}. We rewrite the Frobenius norm as
\begin{align*}
\| \Sigma Z^\prime \Sigma - Z^\prime \|_F
& = \| (\Sigma - I) Z^\prime (\Sigma - I) + (\Sigma - I) Z^\prime + Z^\prime (\Sigma - I) \|_F \\
& \leq \| (\Sigma - I) Z^\prime (\Sigma - I) \|_F + 2 \| (\Sigma - I) Z^\prime \|_F . \numberthis\label{fnorm-diff-1}
\end{align*} 
 The first term can be bounded as: 
\begin{align*}
\| (\Sigma - I) Z^\prime (\Sigma - I) \|_F^2 
&= \tr [ (\Sigma - I ) Z^\prime (\Sigma - I)^2 Z^\prime (\Sigma - I) ] \\
&\leq \tr [ (\Sigma - I)^4 \cdot (Z^\prime)^2 ]\\
&\leq 0.01^2 \cdot \tr [ (\Sigma - I)^4 ]  \\
&= \sum_{i=1}^n (\sigma_i - 1)^4 \\
&\leq 5 \times 10^{-8} , \numberthis\label{fnorm-diff-2} 
\end{align*} 
The first inequality above uses Fact~\ref{fact:gen_Lieb}, the second used the observation that $\| Z' \|_{\op} = \| Z \|_{\op} \leq 0.01$, and the last inequality follows from  \eqref{eq:wtw_sos} and the fact that $\sum_{i=1}^n (\sigma_i - 1)^4 \leq \sum_{i=1}^n (\sigma_i^2 - 1)^2$. Similarly, we can bound the second term as
\begin{align*}
\| (\Sigma - I) Z^\prime \|_F^2 
& = \tr[ (\Sigma - I) (Z^\prime)^2 (\Sigma - I) ]\\
&\leq  \tr [ (\Sigma - I)^2 (Z^\prime)^2 ] \\
&\leq 0.01^2 \cdot \tr [(\Sigma - I)^2 ] ~\leq~ 10^{-7}. \numberthis\label{fnorm-diff-3}
\end{align*} 
It follows from \eqref{wzwt-z-1}, \eqref{fnorm-diff-1} and \eqref{fnorm-diff-3} that 
\begin{align}\label{eq:yata_y}
\sum_{i=1}^n (\lambda(W Z W^\top)_{[i]} - \lambda(Z)_{[i]})^2 \leq 10^{-6}.
\end{align}
Combining \eqref{eq:ynew_yata} and \eqref{eq:yata_y}, we get that $\sum_{i=1}^n (\lambda(Z)_{[i]} -\lambda(Z_{\mathrm{mid}})_{[i]})^2 \leq 10^{-3}$ which establishes \eqref{eq:Z_close_to_Zmid}. This completes the proof of the lemma.
\end{proof}

%%%%%%%%%%%%%%%%%%%%%%%%%%%%%%%%%%%%%%%%%%%%%
%%%%%%%%%%%%%%%%%%%%%%%%%%%%%%%%%%%%%%%%%%%%%

\begin{lemma}[Potential change when $\wt{S}$ changes]\label{lem:tildeS_change_pot}
Given positive definite matrices $S_{\new}, \wt{S} \in \Pdn$, let $\wt{S}_{\new}$ and $r$ be generated during the run of Algorithm~\ref{alg:approx_slack_update} when the inputs are $S_{\new}$ and $\wt{S}$. 
Define the matrices 
$Z_{\mathrm{mid}} = (S_{\new})^{-1/2} \wt{S} (S_{\new})^{-1/2} - I$ and
$Z_{\new} = (S_{\new})^{-1/2} \wt{S}_{\new} (S_{\new})^{-1/2} - I$.
Then we have 
\begin{align*}
\Phi(Z_{\mathrm{mid}}) - \Phi(Z_{\new}) \geq \frac{10^{-4}}{\log n} \sqrt{r} .
\end{align*}
\end{lemma}
\begin{proof}
The setup of the lemma  considers the eigenvalues of $Z$ when $\wt{S}$ changes. 
%Before starting the analysis, we recall that the eigenvalues of $Z$ are all sorted in decreasing order of magnitude. 
For the sake of notational convenience, we define $y = |\lambda(Z_{\mathrm{mid}})|$, the vector of absolute values of eigenvalues of $Z_{\mathrm{mid}} = S_{\new}^{-1/2} \wt{S} S_{\new}^{-1/2} - I$. Recall from Table~\ref{Tab:params} that $\epsilon_S = 0.01$. We consider two cases below.

\medskip
\noindent \textbf{Case 1.} There does not exist an $i \leq n/2$ that satisfies the two conditions $y_{[2i]} < \epsilon_S$ and $y_{[2i]} <(1 - 1/10 \log n) y_{[i]}$. 
In this case, we have $r = n/2$. 
We consider two sub-cases.  
\begin{itemize}
\item Case (a). For all $i \in [n]$, we have $y_{[i]} \geq \epsilon_S$. In this case, we change all $n$ coordinates of $y$, and the change in each coordinate contributes to a potential decrease of at least $\epsilon_S/\sqrt{n}$.
Therefore, we have $\Phi(Z_{\mathrm{mid}}) - \Phi(Z_{\mathrm{new}}) \geq \epsilon_S \sqrt{n} \geq \frac{10^{-4}}{\log n} \sqrt{r}$.

\item Case (b). There exists a minimum index $i \leq n/2$ such that $y_{[2j]} < \epsilon_S$ holds for all $j$ in the range $i \leq j \leq n/2$.
In this case, for all $j$ in the above range, we have that $y_{[2j]} \geq (1 - 1/10 \log n) y_{[j]}$. 
In particular, picking $j = i, 2i, \cdots $ gives
%Then, to check for the second condition being satisfied, it suffices to start checking from the  index $\lceil j/2 \rceil$, which has $y_{\lceil j/2 \rceil} \geq c$. Recall that we are in the case in which we eventually do \emph{not} find an  index $i$ such that $y_{[2i]} \leq (1 - 1/(10\log n)) y_{[i]}$. Therefore, starting at  index $\lceil j/2 \rceil$, we can keep doubling the index to check this condition and be sure that the very last coordinate still doesn't satisfy this condition.
%We thus have the inequality 
\begin{align*}
y_{[n]} \geq y_{[i]} \cdot (1 - 1/(10\log n))^{\lceil \log n \rceil} \geq \epsilon_S/10 .
\end{align*}
%where we get $\log n$ in the exponent because there are at most $\lceil \log n \rceil$ ``doubling jumps'' possible when we have a total of $n$ indices. 
%Then, repeating the argument from the previous sub-case, the total potential change is at least $\frac{1}{\sqrt{n}} \cdot c/10 \cdot n = (c/10)\cdot \sqrt{n}$. So, even though we do have an index that satisfies $y_{[i]} \leq c$, all the indices satisfy $y_{[i]}\geq c/10$, which leads us to the same conclusion as the previous subcase.  
Recalling that our notation $y_{[i]}$ denotes the $i$'th absolute eigenvalue in decreasing order, we use the above inequality and repeat the argument from the previous sub-case to conclude that $\Phi(Z_{\mathrm{mid}}) - \Phi(Z_{\mathrm{new}}) \geq \epsilon_S/10 \cdot \sqrt{n} \geq \frac{10^{-4}}{\log n} \cdot \sqrt{r}$. 
\end{itemize}

\noindent \textbf{Case 2.} There exists an index $i$ for which both the conditions $y_{[2i]} < \epsilon_S$ and $y_{[2i]} <(1 - 1/10 \log n) y_{[i]}$ are satisfied. 
By definition, $r \leq n/2$ is the smallest such index. 
Consider the index $j$ such that for all $j' < j$, we have $y_{[j']} \geq \epsilon_S$  and for all $j' \geq j$, we have $y_{[j]} < \epsilon_S$. 
By the same argument as in Case 1(b), we can prove $y_{[r]} \geq \epsilon_S/10$.
Moreover, $y_{[2r]} < (1 - 1/10 \log n) y_{[r]}$ by definition of $r$. 
Denote by $y^{\mathrm{new}}$ the vector of  magnitudes of the eigenvalues of $Z_{\mathrm{new}}$. 
Since $y^{\mathrm{new}}_{[i]}$ is set to $0$ for each $i \in [2r]$, we have $y^{\mathrm{new}}_{[i]} = y_{[i + 2r]} \leq y_{[i]}$. 
Further, $y_{[2r]} < (1 - 1/10 \log n) y_{[r]}$ implies that for each $i \in [r]$, we have
\begin{align*}
y_{[i]} - y^{\mathrm{new}}_{[i]} \geq \frac{1}{10 \log n} \cdot y_{[r]} \geq \frac{10^{-2} \epsilon_S}{\log n} = \frac{10^{-4}}{\log n} ,
\end{align*}
where $\epsilon_S = 0.01$ by Table~\ref{Tab:params}. 
Therefore, we can bound, from below, the decrease in potential function as
\begin{align*}
\Phi(Z_{\mathrm{mid}}) - \Phi(Z_{\mathrm{new}}) \geq \sum_{i=1}^r \frac{y_{[i]} - y^{\mathrm{new}}_{[i]}}{\sqrt{i}} \geq \frac{10^{-4}}{\log n} \sqrt{r} .
\end{align*}
This finishes the proof of the lemma.
%Consider, as in Case $1b$, the index $\lceil j/2 \rceil$, at which we start searching for the first condition to be satisfied, doubling the index each time it's not. As argued in Case $1b$, the first index $i$ at which this condition is met satisfies $y_{[i]} \geq c (1 - 1/(10 \log n))^{\lceil \log n \rceil} \geq c/10$. The first $2i$ coordinates of $y$ are zeroed out; after rearranging, the potential strictly decreases  since we are moving the smaller values to replace the larger ones. Therefore we can consider the decrease in potential due to the first (after rearranging) $i$ coordinates, and that gives a lower bound on the decrease in potential. The terms that are replacing the first $i$ coordinates are smaller than $y_{[i]} (1 - 1/(10 \log n))$. Therefore the change in potential is at least $c/100\log n$. Since the rank changes by $r$, the number of coordinates that change is also $r$; therefore, $i = r$. The total change in potential is therefore at least 
%\begin{align*}
%\sum_{i=1}^r \frac{1}{\sqrt{i}} \frac{c}{100 \log n}\geq  10^{-4}\sqrt{r}/\log n.
%\end{align*} 
\end{proof} 

%%%%%%%%%%%%%%%%%%%%%%%%%%%%%%%%%%%%%%%%%%
%%%%%%%%%%%%%%%%%%%%%%%%%%%%%%%%%%%%%%%%%%

\begin{proof}[Proof of Theorem~\ref{thm:rankineq}] Recall the definition of the potential function in \eqref{eq:def_pot_sdp} for an error matrix $Z \in \mathbb{S}^{n \times n}$:
\begin{align*}
\Phi(Z) = \sum_{i=1}^n \frac{|\lambda(Z)|_{[i]}}{\sqrt{i}} .
\end{align*}
Let $S^{(i)}$ and $\wt{S}^{(i)}$ be the true and approximate slack matrices in the $i$th iteration of Algorithm~\ref{alg:ss_ipm}. 
Define $Z^{(i)} = (S^{(i)})^{-1/2} \wt{S}^{(i)} (S^{(i)})^{-1/2} - I $ and $Z_{\mathrm{mid}}^{(i)} = (S^{(i+1)})^{-1/2} \wt{S}^{(i)} (S^{(i+1)})^{-1/2} - I$. 
By Lemma~\ref{lem:S_change_pot}, we have that
\begin{align*}
\Phi(Z_{\mathrm{mid}}^{(i)}) - \Phi(Z^{(i)}) \leq \sqrt{\log n} .
\end{align*}
From Lemma~\ref{lem:tildeS_change_pot}, we have the following potential decrease:
\begin{align*}
\Phi(Z_{\mathrm{mid}}^{(i)}) - \Phi(Z^{(i+1)}) \geq \frac{10^{-4}}{\log n} \sqrt{r_i} .
\end{align*}
These together imply that 
\begin{align}\label{eq:potential_drop}
\Phi(Z^{(i+1)}) - \Phi(Z^{(i)}) \leq \sqrt{\log n} - \frac{10^{-4}}{\log n} \sqrt{r_i} . 
\end{align}
We note that $\Phi(Z^{(0)}) = 0$ as we initialized $\wt{S} = S$ in the beginning of the algorithm, and that the potential function $\Phi(Z)$ is always non-negative. The theorem then follows by summing up \eqref{eq:potential_drop} over all $T$ iterations. 
\end{proof} %%% Section 5. Low-rank update

%\newpage
\section{Runtime Analysis}\label{sec-cost}
Our main result of this section is the following bound on the runtime of Algorithm~\ref{alg:ss_ipm}.

\begin{restatable}[Runtime bound]{theorem}{runtime}\label{thm:runtime}
The total runtime of Algorithm~\ref{alg:ss_ipm} for solving an SDP with variable size $n\times n$ and $m$ constraints is at most $O^*\left(\sqrt{n} \left(mn^2 + \max(m,n)^\omega \right) \right)$, where $\omega$ is the matrix multiplication exponent as defined in Definition~\ref{def:omegak}.
\end{restatable}

To prove Theorem~\ref{thm:runtime}, we first upper bound the runtime in terms of fast rectangular matrix multiplication times. The iteration complexity of Algorithm~\ref{alg:ss_ipm} is $T = \wt{O}(\sqrt{n})$. 

%%%%%%%%%%%%%%%%%%%%%%%%%%%%%%%%%
%%%%%%%%%%%%%%%%%%%%%%%%%%%%%%%%%

\begin{lemma}[Total cost]\label{lem:totalcost}
The total runtime of Algorithm~\ref{alg:ss_ipm} over $T$ iterations is upper bounded as
\begin{align}\label{eq:totalcost}
\cost{\mathrm{Total}} \leq O^* \left( \min \left( n \cdot \nnz(A), m n^{2.5} \right ) + \sqrt{n} \max(m,n)^\omega + \sum_{i = 0}^{T} \left(\Tmat(n, m r_i, n) + \Tmat(m, n r_i, m) \right) \right) ,
\end{align} 
where $\nnz(A)$ is the total number of non-zero entries in all the constraint matrices, $r_i$, as defined in Theorem~\ref{thm:rankineq}, is the rank of the update to the approximation slack matrix $\widetilde{S}$ in iteration $i$, and $\omega$ and $\Tmat$ are defined in Definitions~\ref{def:omegak} and ~\ref{def:cmat}, respectively. 
\end{lemma}

\begin{remark}
A more careful analysis can improve the first term in the RHS of \eqref{eq:totalcost} to $\sqrt{n} \cdot \nnz(A)^{1-\gamma} \cdot (mn^2)^\gamma$ for $\gamma = \frac{1}{2 (3 - \omega(1))}$. For the purpose of this paper, however, we will only need the simpler bound given in Lemma~\ref{lem:totalcost}.
\end{remark}

\begin{proof}
The total runtime of Algorithm~\ref{alg:ss_ipm} consists of two parts:
\begin{itemize}
\vspace{-0.2cm}
 \item {\bf Part 1.} The time to compute the approximate Hessian $\wt{H}(y)$ (which we abbreviate as $\wt{H}$) in Line~\ref{lin:start_compute_Hessian}~-~\ref{lin:end_compute_Hessian}.
 \vspace{-0.15cm}
 \item {\bf Part 2.} The total cost of operations other than computing the approximate Hessian. 
\end{itemize}
\vspace{-0.15cm}

{\bf Part 1.}

We analyze the cost of computing the approximate Hessian $\wt{H}$. 

{\bf Part 1a. Initialization.}

We start with computing $\wt{H}$ in the first iteration of the algorithm. Each entry of $\wt{H}$ involves the computation 
\begin{align*}
\wt{H}_{j,k} = \tr \Big[ (\wt{S}^{-1/2} A_j \wt{S}^{-1/2})( \wt{S}^{-1/2} A_k \wt{S}^{-1/2}) \Big].
\end{align*} 
It first costs $O^*(n^\omega)$ to invert $\widetilde{S}$. Then the cost of computing the key module of the approximate Hessian, $\wt{S}^{-1/2} A_j \wt{S}^{-1/2}$ for all $j\in [m]$, is obtained by stacking the matrices $A_j$ together:
\begin{align}\label{eq:cost-H0-1}
\cost{\wt{S}^{-1/2} A_j \wt{S}^{-1/2} \text{ for all } j \in [m]} \leq O(\Tmat(n, mn, n)).
\end{align}
Vectorizing the matrices $\wt{S}^{-1/2}A_j \wt{S}^{-1/2}$ into row vectors of length $n^2$, for each $j \in [m]$, and stacking these rows vertically to form a matrix $B$ of dimensions $m \times n^2$, one observes that $\wt{H} = BB^{\top}$. We therefore have, 
\begin{align}\label{eq:exactcost-H0-2}
\cost{\text{computing } \wt{H} \text{ from } B} \leq O(\Tmat(m, n^2, m)). 
\end{align}  
Combining \eqref{eq:cost-H0-1}, \eqref{eq:exactcost-H0-2}, and the initial cost of inverting $\wt{S}$ gives the following cost for computing $\wt{H}$ for the first iteration: 
\begin{align}\label{eq:time_part_1a}
{\cal T}_{\text{part~1a}} \leq O^*(\Tmat(m, n^2, m) + \Tmat(n, mn, n) + n^\omega).
\end{align}  

{\bf Part 1b. Accumulating low-rank changes over all the iterations}

Once the approximate Hessian in the first iteration has been computed, every next iteration has the approximate Hessian computed using a rank $r_i$ update to the approximate slack matrix $\wt{S}$ (see Line~\ref{lin:low_rank_update} of Algorithm~\ref{alg:approx_slack_update}). If the update from $\wt{S}$ to $\wt{S}_{\new}$ has rank $r_i$, Fact~\ref{fac:woodbury} implies that we can compute, in time $O(n^{\omega + o(1)})$, the $n \times r_i$ matrices $V_+$ and $V_-$ satisfying $\wt{S}_{\new}^{-1} = \wt{S}^{-1} + V_+ V_+^\top - V_- V_-^\top$.  The cost of updating $\wt{H}$ is then dominated by the computation of $\tr[ \wt{S}^{-1/2} A_j V V^{\top} A_k \wt{S}^{-1/2} ]$, where $V \in \R^{n \times r_i}$ is either $V_+$ or $V_-$. We note that
\begin{align}\label{eq:exactcost_Hr_1}
\cost{A_j V \text{ for all } j \in [m]} \leq O^* \left(\min \left( r_i \cdot \nnz(A), mn^2 r_i^{\omega - 2 + o(1)} \right ) \right),
\end{align}
where $\nnz(A)$ is the total number of non-zero entries in all the constraint matrices, and the second term in the minimum is obtained by stacking the matrices $A_j$ together and splitting it and $V$ into matrices of dimensions $r_i \times r_i$. Further, pre-multiplying $\wt{S}^{-1/2}$ with $A_j V$ for all $j\in [m]$ essentially involves computing the matrix product of an $n \times n$ matrix and an $n \times mr_i$ matrix, which, by Definition~\ref{def:cmat}, costs $\Tmat(n, m r_i, n)$. 
This, together with \eqref{eq:exactcost_Hr_1},  gives 
\begin{align}\label{eq:exactcost_Hr_2}
\cost{\wt{S}^{-1/2}A_j V \text{ for all $j\in [m]$}} \leq O^* \left( \Tmat(n, m r_i, n) + \min \left( r_i \cdot \nnz(A), mn^2 r_i^{\omega - 2 + o(1)} \right ) \right).
\end{align}
The final step is to vectorize all the matrices $\wt{S}^{-1/2}A_j V$, for each $j\in [m]$, and stack these vertically to get an $m \times nr_i$ matrix $B$, which gives the update to Hessian to be computed as $BB^{\top}$. This costs, by definition, $\Tmat(m, nr_i, m)$. Combining this with \eqref{eq:exactcost_Hr_2} gives the following run time for one update to the approximate Hessian:

\begin{align}\label{eq:totalcost_ri}
\cost{\text{rank $r_i$ Hessian update}} \leq O^* \left(\Tmat(n, m r_i, n) + \min \left( r_i \cdot \nnz(A), mn^2 r_i^{\omega - 2} \right ) + \Tmat(m, nr_i, m) + n^\omega \right).
\end{align} 
 
Using this bound over all $T = \wt{O}(\sqrt{n})$ iterations, and applying  $\sum_{i = 0}^T \sqrt{r_i} \leq \wt{O}(\sqrt{n})$ from Theorem~\ref{thm:rankineq}, gives
\begin{align}\label{eq:time_part_1b}
{\cal T}_{\text{part~1b}} \leq O^* \left(\min(n \cdot \nnz(A), mn^{2.5}) + \sqrt{n} \cdot n^\omega + \sum_{i=1}^T (\Tmat(n, mr_i, n) + \Tmat(m, n r_i, m)) \right) .
\end{align}

{\bf Combining Part 1a and 1b.}

Combining \eqref{eq:time_part_1a} and \eqref{eq:time_part_1b}, we have
\begin{align}\label{eq:time_part_1}
{\cal T}_{\text{part~1}} 
\leq & ~ {\cal T}_{\text{part~1a}} + {\cal T}_{\text{part~1b}} \notag \\
\leq & ~ O^* \left(\min(n \cdot \nnz(A), mn^{2.5}) + \sqrt{n} \cdot n^\omega + \sum_{i=0}^T (\Tmat(n, mr_i, n) + \Tmat(m, n r_i, m)) \right), 
\end{align}
where we incorporated the bound from \eqref{eq:time_part_1a} into the $i = 0$ case. 

{\bf Part 2.}

Observe that there are four operations performed in Algorithm~\ref{alg:ss_ipm} other than computing $\wt{H}$:

\begin{itemize}
	\vspace{-0.2cm}
\item {\bf Part 2a.} computing the gradient $g_\eta(y)$
\vspace{-0.2cm}
\item {\bf Part 2b.} inverting the approximate Hessian $\wt{H}$
\vspace{-0.2cm}
\item {\bf Part 2c.} updating the dual variables $y_{\new}$ and $S(y_{\new})$
\vspace{-0.2cm}
\item {\bf Part 2d.} computing the new approximate slack matrix $\wt{S}(y_{\new})$
\end{itemize}

{\bf Part 2a.} 
The $i$'th coordinate of the gradient is expressed as $g_{\eta} (y)_i = \eta b_i - \tr [ S^{-1} A_i ]$. The cost per iteration of computing this quantity equals $O(\nnz(A) + n^{\omega+o(1)})$, where the second term comes from inverting the matrix $S$. 

{\bf Part 2b.}
The cost of inverting the approximate Hessian $\wt{H}$ is $O(m^{\omega+o(1)})$ per iteration.

{\bf Part 2c.}
The cost of updating the dual variable $y_{\new} = y - \wt{H}^{-1} g_{\eta_{\mathrm{new}}}(y)$, given $\wt{H}^{-1}$ and $g_{\eta_{\mathrm{new}}}(y)$, is $O(m^2)$ per iteration. The cost of computing the new slack matrix $S_{\new} = \sum_{i \in [m]} (y_{\new})_i A_i - C$ is $O(\nnz(A))$ per iteration.

{\bf Part 2d.}
The per iteration cost of updating the approximate slack matrix $\wt{S}_{\new}$ is $O(n^{\omega + o(1)})$ by Lemma~\ref{lem:approx_slack_update}.

{\bf Combining Part 2a, 2b, 2c and 2d.}

The total cost of operations other than computing the Hessian over the $T = \wt{O}(\sqrt{n})$ iterations is therefore bounded by 
\begin{align}\label{eq:time_part_2}
{\cal T}_{\text{part~2}} 
\leq & ~ {\cal T}_{\text{part~2a}} + {\cal T}_{\text{part~2b}} + {\cal T}_{\text{part~2c}} + {\cal T}_{\text{part~2d}} \notag \\  
\leq & ~ O^*(\sqrt{n} (\nnz(A) + \max(m,n)^\omega)). 
\end{align}

{\bf Combining Part 1 and Part 2.}

Combining \eqref{eq:time_part_1} and \eqref{eq:time_part_2} and using $r_0 = n$ finishes the proof of the lemma. 
\begin{align*}
{\cal T}_{\text{total}}
\leq & ~ {\cal T}_{\text{part~1}} + {\cal T}_{\text{part~2}} \notag \\
\leq & ~ O^* \left( \min \left( n \cdot \nnz(A), m n^{2.5} \right ) + \sqrt{n} \max(m,n)^\omega + \sum_{i = 0}^{T} \left(\Tmat(n, m r_i, n) + \Tmat(m, n r_i, m) \right) \right) .
\end{align*} 

\end{proof}

%%%%%%%%%%%%%%%%%%%%%%%%%%%%%%%
%%%%%%%%%%%%%%%%%%%%%%%%%%%%%%%

\begin{lemma}\label{lem:mnri_cheaper} 
Let $\Tmat$ be as defined in Definition~\ref{def:cmat}. Let $T = \wt{O}(\sqrt{n})$ and $\{ r_1, \cdots, r_T \}$ be a sequence that satisfies 
\begin{align*}
\sum_{i=1}^T \sqrt{r_i} \leq O( T \log^{1.5} n )
\end{align*}
Property I. We have
\begin{align*}
\sum_{i=1}^T \Tmat(m, n r_i, m) \leq O^*(\sqrt{n} \max ( m^\omega, n^{\omega}) + \Tmat(m, n^2, m)) ,
\end{align*}
Property II. We have
\begin{align*}
\sum_{i=1}^T \Tmat(n, m r_i, n) \leq O^*(\sqrt{n} \max ( m^\omega, n^{\omega})  + \Tmat(n, mn, n)) .
\end{align*}
\end{lemma}

\begin{proof}
We give only the proof of Property I, as the proof of Property II is similar. Let $m = n^a$. For each $i\in [T]$, let $r_i = n^{b_i}$, where $b_i \in [0,1]$. Then 
\begin{align}\label{eq:Tmat-m-nri-m}
\Tmat(m, n r_i, m) = \Tmat( n^a , n^{1+b_i} , n^a ) = n^{ a \omega ( ( 1 + b_i ) / a ) + o(1)}.
\end{align}

For each number $k \in \{0,1, \cdots,  \log n \}$, define the set of iterations 
\begin{align*}
I_k = \{i \in [T] ~:~ 2^k \leq r_i \leq 2^{k+1}\}.
\end{align*} 

Then our assumption on the sequence $\{r_1,\cdots, r_T\}$ can be expressed as $\sum_{k=0}^{\log n} | I_k | \cdot 2^{k/2} \leq O( T \log^{1.5} n )$. This implies that for each $k \{0,1,\cdots,\log n\}$, we have $|I_k| \leq O ( T \log^{1.5} n / 2^{k / 2} )$. Next, taking the summation of Eq.~\eqref{eq:Tmat-m-nri-m} over all $i\in [T]$, we have
\begin{align}
\sum_{i=1}^T \Tmat(m,nr_i,m)
= & ~ \sum_{i = 1 }^{T} n^{a \cdot \omega ( (1+ b_i)/a  )} \notag \\
= & ~ \sum_{k = 0}^{ \log n } \sum_{i \in I_k} n^{ a \cdot \omega( ( 1+ b_i ) / a ) } \notag \\
\leq & ~ O(\log n) \cdot \max_{k} \max_{i \in I_k} \frac{T \log^{1.5} n}{2^{k/2} } \cdot n^{ a \cdot \omega( ( 1+ b_i ) / a ) }  \notag \\
\leq & ~ \wt{O}(1) \cdot \max_k \max_{2^k \leq n^{b_i} \leq 2^{k+1}} \frac{\sqrt{n}}{2^{k/2}} \cdot n^{ a \cdot \omega( ( 1+ b_i ) / a ) } \notag \\
\leq & ~ \wt{O}(1) \cdot \max_{b_i \in [0,1]} n^{1/2 - b_i/2 + a \cdot \omega( ( 1+ b_i ) / a ) } , \notag
\end{align}
where the fourth step follows from $T = \wt{O}(\sqrt{n})$. To bound the exponent on $n$ above, we define the function $g$,
\begin{align}\label{eq:mnri_2}
g(b_i) = 1/2 -  b_i / 2 + a \cdot \omega( ( 1+b_i ) / a ).
\end{align}
This  function is convex in $b_i$ due to the convexity of the function $\omega$ (Lemma~\ref{lem:omegaconvex}). Therefore, over the interval $b_i \in [0, 1]$, the maximum of $g$ is attained at one of the end points. We simply evaluate this function at the end points. 

{\bf Case 1.} Consider the case $b_i=0$. In this case, we have $g(0) = 1/2 + a \omega(1/a)$. We consider the following two subcases.  
{\bf Case 1a.} If $a \geq 1$, then we have
\begin{align*}
g(0) = 1/2 + a \cdot \omega(1/a) \leq 1/2 + a \omega(1) = 1/2 + a \omega
\end{align*}

{\bf Case 1b.} If $a \in (0,1)$, then we define $k = 1/a > 1$. It follows from Lemma \ref{lem:sublinearity} and $\omega > 1$, that
\begin{align*}
g(0) = 1/2 + a \cdot \omega(1/a) = 1/2+ \omega(k) / k \leq 1/2 + (k - 1 + \omega) / k \leq 1/2 + \omega .
\end{align*} 
Combining both Case 1a and Case 1b, we have that 
\begin{align*}
n^{g(0)} \leq \max (n^{1/2+ a \omega}, n^{1/2+\omega }) \leq \sqrt{n} \cdot \max (m^\omega, n^\omega).
\end{align*}

{\bf Case 2} Consider the other case of $b_i=1$. In this case, $g(1) = 1/2 - 1/2 +  a \omega(2/a) = a\omega(2/a)$.  

We now finish the proof by combining Case 1 and Case 2 as follows. 
\begin{align*}
\max_{b_i \in [0,1]} n^{1/2 - b_i + a \cdot \omega( ( 1 + b_i ) / a )} 
 \leq \sqrt{n} \max ( m^\omega, n^{\omega}) + n^{a \cdot \omega(2/a)}. 
\end{align*} 
\end{proof}

%%%%%%%%%%%%%%%%%%%%%%%%
%%%%%%%%%%%%%%%%%%%%%%%%
%%%%%%%%%%%%%%%%%%%%%%%%

\begin{proof}[Proof of Theorem \ref{thm:runtime}]

In light of Lemma~\ref{lem:mnri_cheaper}, the upper bound on runtime given in Lemma~\ref{lem:totalcost} can be written as
\begin{align}\label{eq:totalcost_simplify1}
\cost{\mathrm{Total}} \leq O^* \left( \min \left\{ n \cdot \nnz(A), m n^{2.5} \right \} + \sqrt{n} \max(m,n)^\omega + \Tmat(n, m n, n) + \Tmat(m, n^2, m) \right) .
\end{align} 

Combining this with \ref{lem:mn2_trueomega}, we have the following upper bound on the total runtime of Algorithm~\ref{alg:ss_ipm}:
\begin{align*}
\cost{\mathrm{Total}} 
&\leq O^* \left( \min \left\{ n \cdot \nnz(A), m n^{2.5} \right \} + \sqrt{n} \max(m,n)^\omega + \sqrt{n} \left(mn^2  + m^{\omega} \right) \right) \\
&\leq O^* \left(\sqrt{n} \left(mn^2  + \max(m,n)^{\omega} \right)\right) .
\end{align*}
This finishes the proof of the theorem. 
\end{proof}
 %%% Section 6. Runtime Analysis

%\newpage
\section{Comparison with Cutting Plane Method} 

In this section, we prove Theorem~\ref{thm:compare_cutting_plane}, restated below.  

\CompareCuttingPlane*

\begin{remark}
In the dense case with $\nnz(A) = \Theta(mn^2)$, Algorithm~\ref{alg:ss_ipm} is faster than the cutting plane method whenever $m \geq \sqrt{n}$. 
\end{remark}

\begin{proof}[Proof of Theorem~\ref{thm:compare_cutting_plane}]
Recall that the current best runtime of the cutting plane method for solving an SDP \eqref{eq:sdpprimal} is $ \cost{\textrm{CP}} = O^*(m \cdot \nnz(A) + mn^{2.372927} + m^3)$~\cite{lsw15,jlsw20}, where $2.372927$ is the current best upper bound on the exponent of matrix multiplication $\omega$. 
By Lemma \ref{lem:totalcost} and \ref{lem:mnri_cheaper}, we have the following upper bound on the total runtime of Algorithm~\ref{alg:ss_ipm}:
\begin{align*}
%\label{eq:totalcost_simplify2}
\cost{\mathrm{Total}} \leq O^* \left( \min \left\{ n \cdot \nnz(A), m n^{2.5} \right \} + \sqrt{n} \max(m,n)^\omega + \Tmat(n, m n, n) + \Tmat(m, n^2, m) \right) 
\end{align*}
Since $m \geq n$ by assumption, Lemma \ref{lem:mn2_less_nmn} and ~\ref{lem:mn2_less_nmn} further simplify the runtime to
\begin{align}
\label{eq:totalcost_simplify2}
\cost{\mathrm{Total}} \leq O^* \left( \min \left\{ n \cdot \nnz(A), m n^{2.5} \right \} + \sqrt{n} m^\omega + \Tmat(m, n^2, m) \right) 
\end{align}
Note that $\min \left\{ n \cdot \nnz(A), m n^{2.5} \right \} \leq m \cdot \nnz(A) \leq O(\cost{\textrm{CP}})$ and that $\sqrt{n} m^{\omega} = o(m^3) \leq o(\cost{\textrm{CP}})$ since $m \geq n$.
Furthermore, Lemma~\ref{lem:mn2_cheaper} states that $\Tmat(m, n^2, m) = o(m^3 + m n^{2.37}) \leq o(\cost{\textrm{CP}})$.
Since each term on the RHS of \eqref{eq:totalcost_simplify2} is upper bounded by $\cost{\textrm{CP}}$, we make the stated conclusion.
\end{proof} %%% Section 7. Comparison with Cutting Plane

\section{Initialization} \label{sec:init}

\begin{lemma}[Initialization] \label{lem:init}
Consider a semidefinite program as in Definition~\ref{defn:sdpprimal} of dimension $n \times n$ with $m$ constraints, and assume that it has the following properties. 
\begin{enumerate}
\item Bounded diameter: for any $X \succeq 0$ with $\langle A_i , X \rangle = b_i$ for all $i \in [m]$, we have $\norm{X}_\op \leq R$.
\item Lipschitz objective: $\norm{C}_\op \leq L$. 
\end{enumerate}
For any $0 < \delta \leq 1$, the following {\em modified} semidefinite program
\begin{align*}
\max_{\overline{X} \succeq 0} ~ & \langle \overline{C}, \overline{X} \rangle \\
\mathrm{s.t.} ~ & \langle \overline{A}_i, \overline{X} \rangle = \overline{b}_i, \forall i \in [m + 1], 
\end{align*}
where 
\begin{align*}
\overline{A}_i =
\left[ 
\begin{matrix}  
A_i & 0_n & 0_{n} \\
0_n^\top & 0 & 0 \\
0_{n}^\top & 0 & \frac{b_i}{R} - \tr[A_i] 
\end{matrix} 
\right] , \quad \forall i \in [m] ,
\end{align*}
\begin{align*}
\overline{A}_{m+1} &=
\left[ 
\begin{matrix}  
I_n & 0_n & 0_{n} \\
0_n^\top & 1 & 0 \\
0_{n}^\top & 0 & 0
\end{matrix} 
\right] ~,~ 
\overline{b} = 
\left[
\begin{matrix}
\frac{1}{R} b\\
n + 1
\end{matrix}
\right] ~,~
\overline{C} = \left [
\begin{matrix}  
C \cdot \frac{\delta}{L} & 0_n & 0_{n} \\
0_n^\top & 0 & 0 \\
0_{n}^\top & 0 & -1 
\end{matrix} \right] ,
\end{align*}
satisfies the following statements. 
\begin{enumerate}
\item The following are feasible primal and dual solutions: 

\begin{align*}
\overline{X} = I_{n+2} ~,~ \overline{y} = \left[ \begin{matrix} 0_m \\ 1 \end{matrix} \right] ~,~
\overline{S}=
\left[ \begin{matrix}
I_n - C \cdot \frac{\delta }{L} & 0_n & 0 \\
0_n^\top & 1 & 0 \\
0_n^\top & 0 & 1 
\end{matrix}\right] .
\end{align*}

\item For any feasible primal and dual solutions $(\overline{X}, \overline{y}, \overline{S})$ with duality gap at most $\delta^2$, the matrix $\widehat{X} = R \cdot \overline{X}_{[n] \times [n]}$, where $\overline{X}_{[n] \times [n]}$ is the top-left $n \times n$ block submatrix of $\overline{X}$, is an approximate solution to the original semidefinite program in the following sense:
\begin{align*}
\langle C, \widehat{X} \rangle &\geq \langle C, X^* \rangle - L R \cdot \delta ,\\
\widehat{X} &\succeq 0, \\
\sum_{i \in [m]} \left|\langle A_i, \widehat{X} \rangle - b_i \right| &\leq 4 n \delta \cdot (R \sum_{i \in [m]} \norm{A_i}_1 + \norm{b}_1), 
\end{align*}
where $X^*$ is any optimal solution to the original SDP and $\norm{A}_1$ denotes the Schatten $1$-norm of a matrix $A$. 
\end{enumerate}
\end{lemma}

\begin{proof}
For the first result, straightforward calculations show that $\langle \overline{A}_i, \overline{X} \rangle = \overline{b}_i $ for all $i \in [m+1]$, and that $\sum_{i \in [m+1]} \overline{y}_i \overline{A}_i - \overline{S} = \overline{C}$. 
Now we prove the second result. 
Denote $\OPT$ and $\overline{\OPT}$ the optimal values of the original and modified SDP respectively. 
Our first goal is to establish a lower bound for $\overline{\OPT}$ in terms of $\OPT$. 
For any optimal solution $X \in \mathbb{S}^{n \times n}$ of the original SDP, consider the following matrix $\overline{X} \in \mathbb{R}^{(n+2) \times (n+2)}$
\begin{align*}
\overline{X} = 
\left[ \begin{matrix}
\frac{1}{R} X & 0_n & 0_n \\
0_n^\top & n+1 - \frac{1}{R} \tr[X] & 0 \\
0_n^\top & 0 & 0 
\end{matrix} \right] .
\end{align*}
Notice that $\overline{X}$ is a feasible primal solution to the modified SDP, and that 
\begin{align*}
\overline{\OPT} \geq 
\langle \overline{C}, \overline{X} \rangle = \frac{\delta}{LR} \cdot \langle C, X \rangle = \frac{\delta}{LR} \cdot \OPT ,
\end{align*}
where the first step follows because the modified SDP is a maximization problem, and the final step is because $X$ is an optimal solution to the original SDP. 

Given a feasible primal solution $\overline{X} \in \mathbb{R}^{(n+2) \times (n+2)}$ of the modified SDP with duality gap $\delta^2$, we could assume $\overline{X} = \left[ \begin{matrix} \overline{X}_{[n] \times [n]} & 0_n & 0_n \\ 0_n^\top & \tau & 0 \\ 0_n^\top & 0 & \theta \end{matrix} \right]$ without loss of generality, where $\tau, \theta \geq 0$.  
This is because if the entries of $\overline{X}$ other than the diagonal and the top-left $n \times n$ block are not $0$, then we could zero these entries out and the matrix remains feasible and positive semidefinite. 
We thus immediately have $\widehat{X} \succeq 0$. 
Notice that
\begin{align} \label{eq:Xnn}
\frac{\delta}{L} \cdot \langle C, \overline{X}_{[n] \times [n]} \rangle - \theta = \langle \overline{C}, \overline{X} \rangle \geq \overline{\OPT} - \delta^2 \geq \frac{\delta}{LR} \cdot \OPT - \delta^2 .
\end{align} 
Therefore, we can lower bound the objective value for $\overline{X}_{[n] \times [n]}$ in the original SDP as
\begin{align*}
\langle C, \widehat{X} \rangle = R \cdot \langle C, \overline{X}_{[n] \times [n]} \rangle \geq \OPT - LR \cdot \delta ,
\end{align*}
where the last inequality follows from \eqref{eq:Xnn}.
By matrix H\"{o}lder inequality, we have 
\begin{align*}
\frac{\delta}{L} \cdot \langle C, \overline{X}_{[n] \times [n]} \rangle 
& \leq \frac{\delta}{L} \cdot \norm{C}_\op \cdot \tr\left[\overline{X}_{[n] \times [n]} \right] \\
& \leq \frac{\delta}{L} \cdot \norm{C}_\op \cdot \langle \overline{A}_{m+1}, \overline{X} \rangle  \\
& \leq (n+1) \delta ,
\end{align*}
where in the last step follows from $\norm{C}_\op \leq L$ and $b_{m+1} = n+1$. 
We can thus upper bound $\theta$ as 
\begin{align} \label{eq:theta_upper_bound}
\theta 
\leq \frac{\delta}{L} \cdot \langle C, \overline{X}_{[n] \times [n]} \rangle + \delta^2 - 
\frac{\delta}{LR} \cdot \OPT 
\leq (2n + 1) \delta  + \delta^2 \leq 4 n \delta ,
\end{align}
where the first step follows from \eqref{eq:Xnn}, the second step follows from $\OPT \geq - \norm{C}_\op \cdot \norm{X^*}_1 \geq - n L R$ where $\norm{\cdot}_1$ is the Schatten $1$-norm, and the last step follows from $\delta \leq 1 \leq n$. 
Notice that by the feasiblity of $\overline{X}$ for the modified SDP, we have
\begin{align*}
\langle A_i, \overline{X}_{[n] \times [n]} \rangle + (\frac{1}{R} \cdot b_i - \tr[A_i] ) \theta = \frac{1}{R} \cdot b_i .
\end{align*}
This implies that 
\begin{align*}
\left|\langle A_i, \widehat{X} \rangle - b_i \right| = |(b_i - R \cdot \tr[A_i]) \theta| \leq 4 n \delta \cdot (R \norm{A_i}_1 + |b_i|),
\end{align*}
where the final step follows from the upper bound of $\theta$ in \eqref{eq:theta_upper_bound}. 
Summing the above inequality up over all $i \in [m]$ finishes the proof of the lemma. 
\end{proof}

 %%% Section 8. Initilization

\fi

\section*{Acknowledgment} We thank Aaron Sidford for many helpful discussions and Deeksha Adil, Sally Dong, Sandy Kaplan, and Kevin Tian  for useful feedback on the writing. We gratefully acknowledge funding from CCF-1749609, CCF-1740551, DMS-1839116, Microsoft Research Faculty Fellowship, and Sloan Research Fellowship.  Zhao Song is partially supported by Ma Huateng Foundation, Schmidt Foundation, Simons Foundation, NSF, DARPA/SRC, Google and Amazon.

\ifdefined\isfocs

\else
\newpage
\fi
\bibliographystyle{alpha}	
\bibliography{ref}
\newpage
\begin{appendix}
\section{Matrix Multiplication: A Tensor Approach}\label{subsec:tensormatproperty}

The main goal of this section is to rederive, using tensors, some of the technical results from Section~\ref{subsec:matproperty}. In particular, we use tensors to derive upper bounds on the time to perform the following two rectangular matrix multiplication tasks (Lemma~\ref{lem:tensormn2_less_nmn} and ~\ref{lem:tensormn2_trueomega}): 
\begin{itemize}
	\vspace{-0.2cm}
\item Multiplying a matrix of dimensions $m \times n^2$ with one of dimensions $n^2 \times m$,
\vspace{-0.2cm}
\item Multiplying a matrix of dimensions $n \times mn$ with one of dimensions $mn \times n$. 
\end{itemize} 
\vspace{-0.2cm} Our hope is that these techniques will eventually be useful in further improving the results of this paper. 

\subsection{Exponent of matrix multiplication}

We recall two definitions to describe the cost of certain fundamental matrix operations, along with their properties. 

\begin{definition}\label{def:tensorcmat}
Define $\Tmat(n,r,m)$ to be the number of operations needed to compute the product of matrices of dimensions $n \times r$ and $r \times m$.
\end{definition}

\begin{definition}\label{def:tensoromegak} 
We define the function $\omega(k)$ to be the
minimum value such that $\Tmat( n, n^k, n ) = n^{ \omega( k ) + o( 1 ) } $. We overload notation and use $\omega$ to denote the exponent of matrix multiplication (in other words, the cost of multiplying two $n\times n$ matrices is $n^\omega$), and let $\alpha$ denote the dual exponent of matrix multiplication. Thus, we have $\omega(1)= \omega$ and $\omega(\alpha) = 2$. 
\end{definition}

\begin{lemma}[\cite{gu18}] \label{lem:tensoromega2_bound}
We have : 
\begin{itemize}
	\vspace{-0.2cm}
\item  $\omega = \omega(1) \leq 2.372927$, 
\vspace{-0.2cm}
\item  $\omega (1.5) \leq 2.79654$, 
\vspace{-0.2cm}
\item  $\omega (1.75) \leq 3.02159$,
\vspace{-0.2cm}
\item  $\omega(2) \leq 3.251640$. 
\end{itemize}
\end{lemma}

\begin{lemma}[\cite{bcs97,b13}]\label{lem:tensorOrderOfTmat}
For any three positive integers $n,m,r$, we have
\begin{align*}
\Tmat(n,r,m) = O(\Tmat(n,m,r)) = O(\Tmat(m,n,r)).
\end{align*}
\end{lemma}

\subsection{Matrix multiplication tensor}
The rank of a tensor $T$, denoted as $R(T)$, is the minimum number of simple tensors that sum up to $T$. For any two tensors $S = (S_{i,j,k})_{i,j,k}$ and $T = (T_{a,b,c})_{a,b,c}$, we write $S \leq T$ if there exist three matrices $A, B$ and $C$ (of appropriate sizes) such that $S_{i,j,k} = \sum_{a,b,c} A_{i,a} B_{j,b} C_{k,c} T_{a,b,c}$ for all $i,j,k$. For any $i,j,k$, denote $e_{i,j,k}$ the tensor with $1$ in the $(i,j,k)$-th entry, and $0$ elsewhere. 

\begin{definition}[Matrix-multiplication tensor]
For any three positive integers $a,b,c$, we define %$\langle a, b, c \rangle$  
\begin{align*}
\langle a, b , c \rangle := \sum_{i \in [a]} \sum_{j \in [b]} \sum_{k \in [c]} e_{i(b-1)+j,j(c-1)+k,k(a-1) + i}
\end{align*}
to be the matrix-multiplication tensor corresponding to multiplying a matrix of size $a \times b$ with one of size $b \times c$. 
\end{definition}

It's not hard to show that for any $n_i$ and $m_i$ where $i = 1,2,3$, we have
\begin{align*}
\langle n_1 , n_2 , n_3 \rangle \otimes \langle m_1 , m_2 , m_3 \rangle = \langle n_1 m_1, n_2 m_2 , n_3 m_3 \rangle.
\end{align*}
Let $\langle n \rangle = \sum_{i \in [n]} e_{i,i,i}$ be the identity tensor. For any three tensors $S, T_1$ and $T_2$, if $T_1 \leq T_2$, then we have
\begin{align*}
S \otimes T_1 \leq S \otimes T_2.
\end{align*}

\begin{lemma}[Monotonicity of tensor rank, \cite{s91}]
\label{lem:tensormonotonicity}
%Let $R(T)$ denote the tensor rank of $T$. 
Tensor rank is monotone under the relation $\leq $, i.e. if $T_1 \leq T_2$, then we have 
\begin{align*}
R(T_1) \leq R(T_2).
\end{align*}
\end{lemma}

\begin{lemma}[Sub-multiplicity of tensor rank, \cite{s91}]
\label{lem:tensorsubmultiplicity}
%Let $R(T)$ denote the tensor rank of $T$. 
For any tensors $T_1$ and $T_2$, we have
\begin{align*}
R(T_1 \otimes T_2) \leq R(T_1) \cdot R(T_2).
\end{align*}
\end{lemma}

\begin{lemma}\label{lem:tensorequiv_rank_time}
%Let $R(T)$ denote the tensor rank of $T$.
The tensor rank of a matrix multiplication tensor is equal to the cost of multiplying the two correponding sized matrices up to some constant factor, i.e.,
\begin{align*}
R( \langle a , b , c \rangle ) = \Theta( \Tmat(a,b,c) ).
\end{align*}
\end{lemma}

\subsection{Implication of matrix multiplication technique}

%%%%%%%%%%%%%%%%%%%%%%%%%%%%%%
%%%%%%%%%%%%%%%%%%%%%%%%%%%%%%

\begin{lemma}[Sub-linearity] \label{lem:tensorsublinearity}
For any $p \geq q \geq 1$, we have
\begin{align*}
\omega(p) \leq p - q + \omega(q).
\end{align*}
\end{lemma}

\begin{proof}
We have
\begin{align*}
\langle n, n^p, n \rangle = \langle n, n^q , n \rangle \otimes \langle 1, n^{p-q} , 1 \rangle.
\end{align*}
Applying tensor rank on both sides
\begin{align*}
R ( \langle n, n^p, n \rangle ) 
= & ~ R ( \langle n, n^q , n \rangle \otimes \langle 1, n^{p-q} , 1 \rangle ) \\
\leq & ~ R ( \langle n, n^q , n \rangle ) \cdot R( \langle 1, n^{p-q} , 1 \rangle ),
\end{align*}
where the last line follows from Lemma~\ref{lem:tensorsubmultiplicity}. 
Applying Lemma~\ref{lem:tensorequiv_rank_time}, we have
\begin{align*}
\Tmat ( n, n^p ,n ) \leq O(1) \cdot \Tmat( n, n^q , n ) \cdot n^{p-q}
\end{align*}
Using the definition of $\omega(p)$, we have
\begin{align*}
n^{ \omega(p) + o(1) } \leq O(1) \cdot n^{ \omega(q) + o(1) } \cdot n^{p-q}.
\end{align*}
Comparing the exponent on both sides completes the proof.
\end{proof}

%%%%%%%%%%%%%%%%%%%%%%%%%%%%%%
%%%%%%%%%%%%%%%%%%%%%%%%%%%%%%

%Key to our analysis is the following lemma, which
The next lemma establishes the convexity of $\omega(k)$ as a function of $k$. 

\begin{lemma}[Convexity of $\omega(k)$]\label{lem:tensoromegaconvex} 
The fast rectangular matrix multiplication time exponent $\omega(k)$ as defined in Definition~\ref{def:tensoromegak} is convex in $k$.
\end{lemma}

\begin{proof}
Let $k = \alpha \cdot p + (1-\alpha) \cdot q$ for $\alpha \in (0,1)$.
We have
\begin{align*}
\langle n , n^k, n \rangle = \langle n^{\alpha} , n^{\alpha \cdot p} , n^{\alpha} \rangle \otimes \langle n^{1-\alpha} , n^{(1-\alpha) p} , n^{1-\alpha} \rangle.
\end{align*}
Applying the tensor rank on both sides,
\begin{align*}
R( \langle n , n^k, n \rangle ) 
= & ~ R ( \langle n^{\alpha} , n^{\alpha \cdot p} , n^{\alpha} \rangle \otimes \langle n^{1-\alpha} , n^{(1-\alpha) p} , n^{1-\alpha} \rangle ) \\
\leq & ~ R ( \langle n^{\alpha} , n^{\alpha \cdot p} , n^{\alpha} \rangle ) \cdot R( \langle n^{1-\alpha} , n^{(1-\alpha) p} , n^{1-\alpha} \rangle ),
\end{align*}
where the last line follows from Lemma~\ref{lem:tensorsubmultiplicity}. 
By Lemma~\ref{lem:tensorequiv_rank_time}, we have
\begin{align*}
\Tmat( n, n^k , n ) \leq O(1) \cdot \Tmat( n^{\alpha} , n^{\alpha p} , n^{\alpha} ) \cdot \Tmat( n^{1-\alpha} , n^{(1-\alpha) p } , n^{1-\alpha} )
\end{align*}
By definition of $\omega (\cdot)$, we have
\begin{align*}
n^{\omega(k) + o(1)} \leq O(1) \cdot n^{ \alpha \cdot \omega(p) } \cdot n^{ (1-\alpha) \omega(1-p)}.
\end{align*}
By comparing the exponent, we know that
\begin{align*}
\omega(k) \leq \alpha \cdot \omega( p ) + ( 1 - \alpha ) \cdot \omega ( 1 - p ) .
\end{align*}
\end{proof}

%%%%%%%%%%%%%%%%%%%%%%%%%%%%%%
%%%%%%%%%%%%%%%%%%%%%%%%%%%%%%

\begin{lemma} \label{lem:tensorhlk_less_hkl}
Let $\Tmat$ be defined as in Definition~\ref{def:tensorcmat}. Then for any positive integers $a,b,c$ and $k$, we have
\begin{align*}
\Tmat(a, b k, c) \leq O(\Tmat(a k , b , c k)) .
\end{align*}
\end{lemma}

\begin{proof}
Notice that 
\begin{align*}
\langle 1 , k , 1 \rangle \leq \langle k , 1 , k \rangle.
\end{align*}
Therefore, we have 
\begin{align*}
\langle a , b k , c \rangle 
= & ~ \langle a , b , c \rangle \otimes \langle 1 , k , 1 \rangle \\
\leq & ~ \langle a , b , c \rangle \otimes \langle k , 1 , k \rangle \\
= & \langle a k , b , c k \rangle.
\end{align*}
It then follows from Lemma~\ref{lem:tensormonotonicity} that 	
\begin{align*}
R( \langle a , b k , c \rangle ) \leq R( \langle a k , b , c k \rangle ).
\end{align*}

Finally, using Lemma~\ref{lem:tensorequiv_rank_time} gives
\begin{align*}
\Tmat(a, b k, c) \leq O(\Tmat(a k , b , c k)).
\end{align*}
Thus we complete the proof.
\end{proof}

%%%%%%%%%%%%%%%%%%%%%%%%%%%%%%
%%%%%%%%%%%%%%%%%%%%%%%%%%%%%%

\subsection{General bound on $\Tmat(n, mn, n)$ and $\Tmat(m, n^2, m)$}

%%%%%%%%%%%%%%%%%%%%%%%%%%%%%%
%%%%%%%%%%%%%%%%%%%%%%%%%%%%%%

\begin{lemma}\label{lem:tensormn2_less_nmn}
Let $\Tmat$ be defined as in Definition~\ref{def:tensorcmat}. \\
If $m \geq n$, then we have
\begin{align*}
\Tmat(n, m n, n) \leq O(\Tmat(m, n^2, m)) .
\end{align*}
If $m \leq n$, then we have
\begin{align*}
\Tmat(m, n^2, m) \leq O(\Tmat(n, m n, n)) .
\end{align*}
\end{lemma}

\begin{proof}
We only prove the case of $m \geq n$, as the other case where $m < n$ is similar. 
This is an immediate consequence of Lemma~\ref{lem:tensorhlk_less_hkl} by taking $a = c = n$, $b = n^2$, and $k = \lfloor m/n \rfloor$, where $k$ is a positive integer because $m \geq n$.
\end{proof}

In the next lemma, we derive upper bounds on the term $\Tmat(m, n^2, m)$ when $m \geq n$ and $\Tmat(n, mn, n)$ when $m < n$, which is crucial to our runtime analysis. 

%%%%%%%%%%%%%%%%%%%%%%%%%%%%%%
%%%%%%%%%%%%%%%%%%%%%%%%%%%%%%

\begin{lemma}\label{lem:tensormn2_trueomega} 
Let $\Tmat$ be defined as in Definition~\ref{def:tensorcmat} and $\omega$ be defined as in Definition~\ref{def:tensoromegak}. \\
Property I. We have 
\begin{align*}
\Tmat(n, mn, n) \leq O (m n^{\omega + o(1)}) .
\end{align*}
Property II. We have
\begin{align*}
\Tmat(m, n^2, m) \leq O \left(\sqrt{n} \left(mn^2  + m^{\omega} \right) \right) .
\end{align*}
%for the true value of $\omega(1)$. 
\end{lemma}
\begin{proof}

{\bf Property I.}

Since
\begin{align*}
\langle n , mn , n \rangle = \langle n , n , n \rangle \otimes \langle 1, m , 1 \rangle.
\end{align*}

Applying the tensor rank on both sides, we have
\begin{align*}
R( \langle n , mn , n \rangle ) 
= & ~ R ( \langle n , n , n \rangle \otimes \langle 1, m , 1 \rangle ) \\
\leq & ~ R( \langle n , n , n \rangle ) \cdot R( \langle 1, m , 1 \rangle  )
\end{align*}
Thus, we complete the proof.  

{\bf Property II.}

Let $m = n^a$, where $a \in (0,\infty)$. We have
\begin{align*}
\langle m , n^2 , m \rangle
= & ~ \langle n^a , (n^a)^{2/a} , n^a \rangle
\end{align*}
It implies that 
\begin{align*}
\Tmat(m, n^2, m) = n^{a \cdot \omega(2/a) + o(1)}
\end{align*}

The Property II is then an immediate consequence of the following inequality, which we prove next: 
\begin{align*}
\omega(2/a) < \max(1 + 2.5/a, \omega(1) + 0.5/a) \qquad \forall a \in (0,\infty) .
\end{align*}
Define $b = 2/a \in (0, \infty)$. 
Then the above desired inequality can be expressed in terms of $b$ as 
\begin{align}\label{tensortbd-lines}
\omega(b) < \max( 1 + 5b/4, \omega(1) + b/4) \qquad \forall b \in (0, \infty) .
\end{align}
Notice that the RHS of \eqref{tbd-lines} is a maximum of two linear functions of $b$ and these intersect at $b^* = \omega(1) - 1$. 
By the convexity of $\omega({}\cdot{})$ as proved in  Lemma~\ref{lem:tensoromegaconvex}, it suffices to verify \eqref{tbd-lines} at the endpoints $b \rightarrow 0$, $b \rightarrow \infty$ and $b = b^*$.
%check the endpoints of the interval of maximization and the point of intersection of the two lines, that is, the points $b = 1$, $b = 2$, and $b = b^*$. 
In the case where $b = \delta$ for any $\delta < 1$, \eqref{tbd-lines} follows immediately from the observation that $\omega(\delta) < \omega(1)$. 
For the case $b \rightarrow \infty$, by Lemma \ref{lem:tensoromega2_bound} we have $\omega(2) \leq 3.252$. 
It then follows from Lemma \ref{lem:tensorsublinearity} that for any $b > 2$, we have 
\begin{align*}
\omega(b) \leq b - 2 + \omega(2) \leq 1 + 5b / 4 .
\end{align*} 
The final case is where $b = b^* = \omega(1) - 1$, for which \eqref{tbd-lines} is equivalent to 
\begin{align}\label{tensortbd-bstar}
\omega(\omega(1) - 1) < 5\omega(1)/4 - 1/4.
\end{align}
%To prove \eqref{tbd-bstar}, we define $t^* = \omega(1) - 2$ and 
By Lemma \ref{lem:tensoromega2_bound}, we have that $\omega(1) - 2 \in [0, 0.372927]$.
Then to prove \eqref{tensortbd-bstar}, it is sufficient to show that 
\begin{align}\label{tensortbd-t}
\omega(t + 1) < 5t/4 + 9/4 \qquad \forall t \in [0, 0.372927] .
\end{align}
By the convexity of $\omega({}\cdot{})$ as proved in Lemma \ref{lem:tensoromegaconvex} and the upper bound of $\omega(2) \leq 3.251640$ in Lemma \ref{lem:tensoromega2_bound}, we have for $k \in [1, 2]$, 
\begin{align*}
\omega(k) \leq \omega(1) + (k-1) \cdot (3.251640 - (t+2)) = t + 2 + (k-1) \cdot (1.251640 - t).
\end{align*} 
In particular, using this inequality for $k = t+1$, we have 
\begin{align*}
	\omega(t + 1) - 5t/4 - 9/4 &\leq (t + 2) + t\cdot (1.251640 - t) - 5t/4 - 9/4 \\
	&= -t^2 + 1.00164 t - 1/4,
\end{align*} 
which is negative on the entire interval $[0,0.372927]$. 
This establishes \eqref{tensortbd-t} and finishes the proof of the lemma. 
\end{proof}

\end{appendix}
\end{document}